\documentclass[12pt]{amsart}

\usepackage{amsfonts, amssymb, amsmath, latexsym,  mathtools, amscd}
\usepackage{hyperref,graphicx}
\usepackage[dvipsnames]{xcolor}
\usepackage{url}
\usepackage{enumitem}
\usepackage{wrapfig}
\usepackage{ulem}  

\DeclareSymbolFont{bbold}{U}{bbold}{m}{n}
\DeclareSymbolFontAlphabet{\mathbbold}{bbold}
\newtheorem{Theorem}{Theorem}[section]
\newtheorem{Proposition}[Theorem]{Proposition}
\newtheorem{Lemma}[Theorem]{Lemma}
\newtheorem{Corollary}[Theorem]{Corollary}

\theoremstyle{remark}
\newtheorem{Remark}[Theorem]{Remark}
\newtheorem{Example}[Theorem]{Example}

\newcommand{\CC}{{\mathbb C}}
\newcommand{\NN}{{\mathbb N}}
\newcommand{\PP}{{\mathbb P}}
\newcommand{\RR}{{\mathbb R}}
\newcommand{\TT}{{\mathbb T}}
\newcommand{\ZZ}{{\mathbb Z}}

\newcommand{\calA}{{\mathcal A}}

\newcommand{\calE}{{\mathcal E}}
\newcommand{\calF}{{\mathcal F}}
\newcommand{\calG}{{\mathcal G}}
\newcommand{\calN}{{\mathcal N}}
\newcommand{\calV}{{\mathcal V}}

\DeclareMathOperator\conv{conv}
\newcommand{\vol}{{\rm vol}}
\newcommand{\Var}{{\it  Var}}
\newcommand{\bfe}{{\bf  e}}
\newcommand{\bfz}{{\bf  0}}

\newcommand{\defcolor}[1]{{\color{blue}#1}}
\newcommand{\demph}[1]{\defcolor{{\sl #1}}}

\title{Critical points of Discrete Periodic Operators}

\author{M.~Faust}
\address{Matthew Faust, Department of Mathematics,
         Texas A\&M University, College Station, Texas 77843,  USA}
\email{mfaust@tamu.edu}
\urladdr{https://mattfaust.github.io}
\author{F.~Sottile}
\address{Frank Sottile, Department of Mathematics,
         Texas A\&M University, College Station, Texas 77843,  USA}
\email{sottile@tamu.edu}
\urladdr{https://franksottile.github.io}
\thanks{Research supported in part by Simons Collaboration Grant for Mathematicians 636314 and NSF grants DMS-2246031, DMS-2201005,
  DMS-2052572, and DMS-2000345.} 
\begin{document}

\begin{abstract}
  We study the spectra of operators on periodic graphs using methods from combinatorial algebraic geometry.
  Our main result is a bound on the number of complex critical points of  the Bloch variety,
  together with an effective criterion for when this bound is attained.
  We show that this criterion holds for $\ZZ^2$- and $\ZZ^3$-periodic graphs 
  with sufficiently many edges and use our results to
  establish the spectral edges conjecture for some $\ZZ^2$-periodic graphs.
\end{abstract}
\maketitle

%
\section*{Introduction}
The spectrum of a $\ZZ^d$-periodic self-adjoint discrete operator $L$ consists of intervals in $\RR$.
Floquet theory reveals that the spectrum is the image of the coordinate projection to $\RR$ of the \demph{Bloch variety}
(also known as the dispersion relation), an algebraic hypersurface in
$(S^1)^d\times\RR$.
This coordinate projection defines a function $\lambda$ on the Bloch variety, which is our main object of study.

When the operator is discrete, the complexification of the Bloch variety is an algebraic variety in $(\CC^\times)^d\times\CC$.
Thus techniques from algebraic geometry and related areas may be used to address some questions in spectral theory.
In the 1990's Gieseker, Kn\"orrer, and Trubowitz~\cite{GKT} used algebraic geometry to study the Schr\"odinger operator on the
square lattice 
$\ZZ^2$ with a periodic potential and established a number of results, including Floquet isospectrality,
the irreducibility of its Fermi varieties, and determined the density of states.
Recently there has been a surge of interest in using algebraic methods in spectral theory.
This includes investigating the irreducibility of Bloch and Fermi varieties~\cite{FLM,FLM23,LiSh,Liu22}, Fermi isospectrality~\cite{Liu+},
density of states~\cite{Kravaris}, and extrema and critical points of the projection $\lambda$ on
Bloch varieties~\cite{Berk,DKS,Liu22}.
We use techniques from combinatorial algebraic geometry and geometric combinatorics~\cite{GBCP} to study critical points of the function
$\lambda$ 
on the Bloch variety of a discrete periodic operator.
We now discuss motivation and sketch our results.
Some background on spectral theory is sketched in Section~\ref{S:one}, and Section~\ref{Sec:AG} gives some background from algebraic geometry.

An old and widely believed conjecture in mathematical physics concerns the structure of the Bloch variety near the edges of the spectral bands.
Namely, that for a sufficiently general operator $L$ (as defined in Section~\ref{S:oneone}),
the extrema of the band functions $\lambda_j$ on the Bloch variety are nondegenerate in that
their Hessians are nondegenerate quadratic forms.
This \demph{spectral edges nondegeneracy conjecture} is stated in~\cite[Conj.\ 5.25]{KuchBAMS}, and it also appears
in~\cite{CdV,KuchBook,Nov81,Nov83}. 
Important notions, such as effective mass in solid state physics, the Liouville property, Green's
function asymptotics, Anderson localization, homogenization, and many other assumed properties in physics, depend upon this conjecture.

The spectral edges conjecture states that for generic parameters, each extreme value is attained by a single band,
the extrema are isolated, and the extrema are nondegenerate. 
We discuss progress for discrete operators on periodic graphs.
In 2000, Klopp and Ralston~\cite{RP} proved that for Laplacians with generic potential each extreme value is attained by a single band.
In 2015, Filonov and Kachkovskiy~\cite{FILI} gave a class of two-dimensional operators for which the extrema are isolated.
They also show~\cite[Sect.\ 6]{FILI} that the spectral edges conjecture may fail for
a Laplacian with general potential, which does not have generic parameters in the sense of Section~\ref{S:oneone}.
Most recently, Liu~\cite{Liu22} proved that the extrema are isolated for the Schr\"{o}dinger operator acting on the square lattice.

We consider a property which implies the spectral edges nondegeneracy conjecture:
A family of operators has the \demph{critical points property} if for almost all operators in the family, all critical points of the
function $\lambda$ (not just the extrema) are nondegenerate. 
Algebraic geometry was used in~\cite{DKS} to prove the following dichotomy:
For a given algebraic family of discrete periodic operators, either the critical points property holds for that family, or almost  all
operators in the family have Bloch varieties with degenerate critical points.

In~\cite{DKS}, this dichotomy was used to establish the critical points property for the family of Laplace-Beltrami difference operators on the
$\ZZ^2$-periodic diatomic graph of Figure~\ref{F:denseDimer}. 
Bloch varieties for these operators were shown to have at most 32 critical points.
A single example was computed to have 32 nondegenerate critical points.
Standard arguments from algebraic geometry (see Section~\ref{S:final}) implied that, for this family, the critical points property, and 
therefore also the spectral edges nondegeneracy conjecture, holds.

We extend part of that argument to operators on many periodic graphs.
Let $L$ be a discrete operator on a $\ZZ^d$-periodic graph $\Gamma$ (see Section~\ref{S:one}).
Its (complexified) Bloch variety is a hypersurface in the product $(\CC^\times)^d\times\CC$ of a complex torus and the complex line
defined by a Laurent polynomial $D(z,\lambda)$.
The last coordinate $\lambda$, corresponding to projection onto the spectral axis, is the function on the Bloch variety
whose critical points we study.
Accordingly, we will call critical points of the function $\lambda$ on the Bloch variety ``critical points of the Bloch variety.''
One contribution of this paper is to shift focus from spectral band functions $\lambda_j$ defined on a compact torus to a
  global function on the complex Bloch variety.
  Another is to use the perspective of nonlinear optimization to address a
question concerning the spectrum of a discrete periodic operator.

We state our first result.
Let $\Gamma$ be a connected $\ZZ^d$-periodic graph (as in Section~\ref{S:oneone}).
Fix a fundamental domain \defcolor{$W$} for the $\ZZ^d$-action on the vertices of $\Gamma$.
The support \defcolor{$\calA(\Gamma)$} of $\Gamma$ records the local connectivity between translates of the fundamental domain.
It is the set of $a \in \ZZ^d$ such that $\Gamma$ has an edge with endpoints in both $W$ and 
$a{+}W$.
\medskip

\noindent{\bf Theorem A.}  \ {\it
 The function $\lambda$ on the Bloch variety of a discrete operator on $\Gamma$ has at most
  \[
      d!\,|W|^{d+1}\vol(\calA(\Gamma))
  \]
 isolated critical points.
 Here, $\vol(\calA(\Gamma))$ is the Euclidean volume of the convex hull of $\calA(\Gamma)$.}\medskip

This bound uses an outer approximation for the Newton polytope of $D(z,\lambda)$ (see Lemma~\ref{L:genericNP}) and a study of the equations
defining critical points of the function $\lambda$ on the Bloch variety, called the \demph{critical point equations}~\eqref{Eq:CPE}.
Corollary~\ref{C:Bound} is a strengthening of Theorem A.
When the bound is attained all critical points are isolated.\medskip

\noindent{\bf Example.}
  We illustrate Theorem A on the example from~\cite[Sect.\ 4]{DKS}.
  Figure~\ref{F:denseDimer} shows a periodic graph $\Gamma$ with $d=2$ whose fundamental domain $W$ has 
  \begin{figure}[htb]
  \centering
  \includegraphics{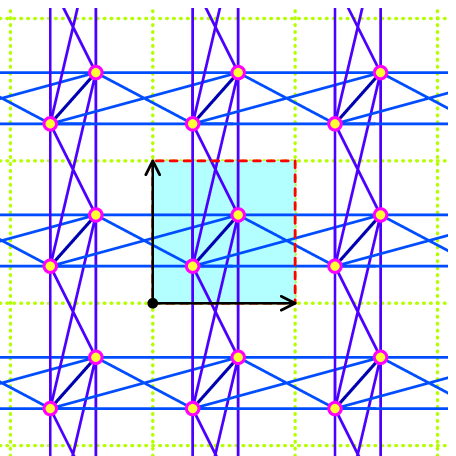}  
  \qquad\qquad
  \raisebox{14.5pt}{\includegraphics{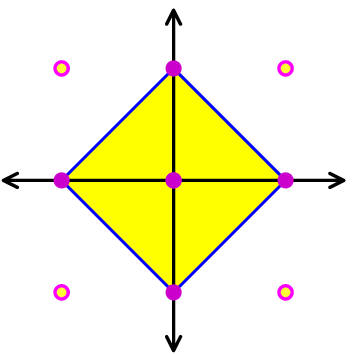}}           
  \caption{A dense periodic graph $\Gamma$ with the convex hull of $\calA(\Gamma)$.}\label{F:denseDimer}
  \end{figure}
  two vertices and its support $\calA(\Gamma)$ consists of the columns of the matrix
  $(\begin{smallmatrix}0&1&0&-1&0\\0&0&1&0&-1\end{smallmatrix})$.
  Figure~\ref{F:denseDimer} also displays the convex hull of $\calA(\Gamma)$.
  As $|W|=2$ and $\vol(\calA(\Gamma))=2$, Theorem A implies that any Bloch variety for an operator on $\Gamma$ has at most 
  \[
  d!\,|W|^{d+1}\vol(\calA(\Gamma))\ =\ 2!\cdot 2^{2+1}\cdot 2=32
  \]
  critical points, which is the bound demonstrated in~\cite{DKS}.
  \hfill$\diamond$\medskip
  
  The bound of Corollary~\ref{C:Bound} arises as follows.
  There is a natural compactification of $(\CC^\times)^d \times \CC$ by a projective toric variety $X$ associated to the Newton
  polytope, P, of $D(z,\lambda)$~\cite[Ch.~5]{GKZ}. 
  The critical point equations become linear equations on $X$ whose number of solutions is the degree of $X$.
  By Kushnirenko's Theorem~\cite{Kushnirenko}, this degree is the normalized volume of $P$, $(d{+}1)!\vol(P)$.
  This bound is attained exactly when there are no solutions at infinity, which is the set 
  $\defcolor{\partial X}\vcentcolon=X \smallsetminus\bigl( (\CC^\times)^d \times \CC\bigr)$ of points added in the compactification. 

  The compactified Bloch variety is a hypersurface in $X$.
  A \demph{vertical face} of $P$ is one that contains a segment parallel to the $\lambda$-axis.
  Corollary~\ref{C:SingThmB} shows that when $P$ has no vertical faces, any solution on $\partial X$ to the critical point equations  is a
  singular point of the intersection of this hypersurface with $\partial X$.
  We state a simplified version of Corollary~\ref{C:SingThmB}.\medskip

\noindent{\bf Theorem B.}  \ {\it
  If $P$ has no vertical faces, then the bound of Corollary~\ref{C:Bound} is attained exactly when the compactified Bloch variety is smooth
  along $\partial X$.}\medskip 

We give a class of graphs whose typical Bloch variety is smooth at infinity and whose Newton polytopes have no vertical faces.
A periodic graph $\Gamma$ is \demph{dense} if it has every possible edge,
given its support $\calA(\Gamma)$ and fundamental domain
$W$ (see Section~\ref{S:C}).
The following is a consequence of Corollary~\ref{C:SingThmB} and Theorem~\ref{Th:denseDense}.  \medskip

\noindent{\bf Theorem C.} \ {\it
  When $d=2$ or $3$ the Bloch variety of a generic operator on a dense periodic graph is smooth along $\partial X$, its Newton
  polytope has no vertical faces, and the bound of \textbf{Theorem A} is attained.}\medskip 

Theorem C is an example of a recent trend in applications of algebraic geometry in which a highly structured
optimization problem is shown to unexpectedly achieve a combinatorial bound on the number of critical points.
A first instance was~\cite{DKS}, which inspired~\cite{EDD} and~\cite{LNRW}.

Section~\ref{S:one} presents background on the spectrum of an operator on a periodic graph, and formulates our goal to bound the number of
critical points of the function $\lambda$ on the Bloch variety.
At the beginning of Section~\ref{S:B}, we recast extrema of the spectral band functions using the language of constrained optimization.
Theorems A, B, and C are proven in Sections~\ref{S:A}, \ref{S:B}, and \ref{S:C}.
In Section~\ref{S:final}, we use these results to prove the spectral edges conjecture for operators on three periodic graphs.

%
\section{Operators on periodic graphs}\label{S:one}

Let $d$ be a positive integer.
We write $\defcolor{\CC^\times}\vcentcolon=\CC\smallsetminus\{0\}$ for the multiplicative group of nonzero complex numbers and 
$\defcolor{\TT}\vcentcolon=\{z\in\CC^\times\mid |z|=1\}$ for its maximal compact subgroup.
Note that if $z\in\TT$, then $\overline{z}=z^{-1}$.
We write edges of a graph as pairs, $(u,v)$ with $u,v$ vertices, and understand that $(u,v)=(v,u)$.

\subsection{Operators on periodic graphs}\label{S:oneone}
For more, see~\cite[Ch.\ 4]{BerkKuch}.
A (\defcolor{$\ZZ^d$}-)\demph{periodic graph} is a simple (no multiple edges or loops) connected undirected graph
$\Gamma$ with a free cocompact action of $\ZZ^d$. 
Thus $\ZZ^d$ acts freely on the vertices, \defcolor{$\calV(\Gamma)$}, and edges,
\defcolor{$\calE(\Gamma)$}, of $\Gamma$ preserving incidences, and $\ZZ^d$ has finitely many orbits
on each of $\calV(\Gamma)$ and $\calE(\Gamma)$.
Figure~\ref{F:first_graphs} 
\begin{figure}[htb]
 \centering
  \includegraphics{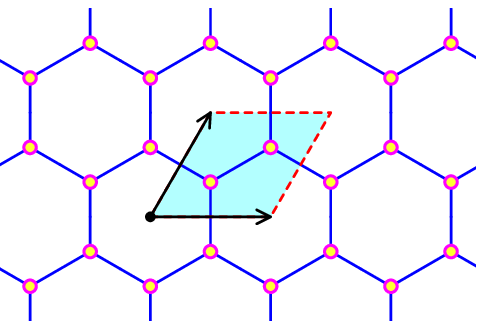}
  \qquad \ 
  \includegraphics{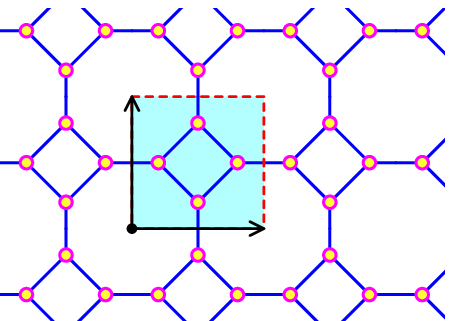}

  \caption{Two $\ZZ^2$-periodic graphs.}
  \label{F:first_graphs}
\end{figure}
shows two $\ZZ^2$-periodic graphs.
One is the honeycomb lattice and the other is an abelian cover of $K_4$, the complete graph on four vertices.

It is useful but not necessary to consider $\Gamma$ immersed in $\RR^d$ so that $\ZZ^d$ acts on $\Gamma$ via translations.
The graphs in Figure~\ref{F:first_graphs} are each immersed in $\RR^2$, and for each we show two independent vectors that generate the
$\ZZ^2$-action. 

Choose a fundamental domain for this $\ZZ^d$-action whose boundary does not contain a vertex of
$\Gamma$. 
In Figure~\ref{F:first_graphs}, we have shaded the fundamental domains.
Let \defcolor{$W$} be the vertices of $\Gamma$ lying in the fundamental domain.
Then $W$ is a set of representatives of $\ZZ^d$-orbits of $\calV(\Gamma)$.
Every $\ZZ^d$-orbit of edges contains one or two edges incident on vertices in $W$.
An edge incident on $W$ has the form $(u,a{+}v)$ for some $u,v\in W$ and $a\in\ZZ^d$.
(If $a=0$, then $u\neq v$ as $\Gamma$ has no loops, and there are no restrictions when $a\neq 0$.)
The support $\defcolor{\calA(\Gamma)}$ of $\Gamma$ is the set of $a\in\ZZ^d$ such that $(u,a{+}v)\in\calE(\Gamma)$ for some
$u,v\in W$. 
This finite set depends on the choice of fundamental domain and it is centrally symmetric in that $\calA(\Gamma)=-\calA(\Gamma)$.
As $\Gamma$ is connected, the $\ZZ$-span of $\calA(\Gamma)$ is $\ZZ^d$.
For both graphs in Figure~\ref{F:first_graphs}, this set consists of the columns of the matrix
$(\begin{smallmatrix}0&1&0&-1&\hspace{3pt}0\\0&0&1&\hspace{3pt}0&-1\end{smallmatrix})$.

A labeling of $\Gamma$ is a pair of functions $\defcolor{e}\colon\calE(\Gamma)\to\RR$ (edge weights) and
$\defcolor{V}\colon\calV(\Gamma)\to\RR$ (potential) that is $\ZZ^d$-invariant (constant on orbits).
The set of labelings is the finite-dimensional vector space $\RR^E\times\RR^W$, where $E$ is the set of orbits on $\calE(\Gamma)$.
Given a labeling $c=(e,V)$, we have the discrete operator \defcolor{$L_c$} acting
on functions $f$ on $\calV(\Gamma)$.
Then $L_c(f)$ is defined by its value at $u\in\calV(\Gamma)$, 
\[
L_c(f)(u)\ \vcentcolon=\ V(u)f(u)\ +\ \sum_{(u,v)\in \calE(\Gamma)} e_{(u,v)}(f(u)-f(v))\,.
\]
We call $L_c$ a \demph{discrete periodic operator} on $\Gamma$, and may often omit the subscript $c$.
It is a bounded self-adjoint operator on the Hilbert space $\ell^2(\Gamma)$ of square-summable functions on $\calV(\Gamma)$, and
has real spectrum.

\subsection{Floquet theory}
As the action of $\ZZ^d$ on $\Gamma$ commutes with the operator $L$, we may apply the Floquet transform, which
reveals important structure of its spectrum.
References for this \demph{Floquet theory} include~\cite{BerkKuch, KuchBook, KuchBAMS}.

The Floquet (Fourier) transform is a linear isometry $\ell^2(\Gamma)\xrightarrow{\,\sim\,}L^2(\TT^d,\CC^W)$,
from $\ell^2(\Gamma)$ to square-inte\-grable 
functions on $\defcolor{\TT^d}$, the compact torus, with values in the vector space $\CC^W$.
The torus $\TT^d$ is the group of unitary characters of $\ZZ^d$.
For $z\in\TT^d$ and $a\in\ZZ^d$, the corresponding character value is the Laurent monomial 
 \[
   \defcolor{z^a}\ \vcentcolon=\ z_1^{a_1} z_2^{a_2}\dotsb z_d^{a_d}\,.
 \]
The Floquet transform \defcolor{$\hat{f}$} of a function $f$ on $\calV(\Gamma)$ is a function on
$\TT^d\times\calV(\Gamma)$ such that for $z\in\TT^d$ and $u\in\calV(\Gamma)$,
 \begin{equation}\label{Eq:FloquetTransform}
   \hat{f}(z,a{+}u)\ =\ z^a \hat{f}(z,u)\qquad\mbox{ for }a\in\ZZ^d\,.
 \end{equation}
Thus $\hat{f}$ is determined by its values at the vertices $W$ in the fundamental domain.

Let $\hat{f}\in L^2(\TT^d,\CC^W)$.
Then for $u\in W$, $\hat{f}(u)$ is a function on $\TT^d$.
The action of the operator $L$ on the Floquet transform $\hat{f}$ is given by the formula 
 \begin{equation}\label{Eq:Laurent_operator}
   L(\hat{f})(u)\  =\ V(u)\hat{f}(u)\ +\ \sum_{(u,a+v)\in\calE(\Gamma)} e_{(u,a+v)} \bigl(\hat{f}(u) - z^a \hat{f}(v)\bigr)\,,
 \end{equation}
as $\hat{f}(a{+}v)=z^a\hat{f}(v)$.
The exponents $a$ which appear lie in the support $\calA(\Gamma)$ of $\Gamma$.
The simplicity of this expression is because $L$ commutes with the $\ZZ^d$-action.

Thus in the standard basis for $\CC^W$, the operator $L$ becomes multiplication by a square matrix whose rows and columns
are indexed by elements of $W$.
Writing \defcolor{$\delta_{u,v}$} for the Kronecker delta function, the matrix entry in position $(u,v)$ is the function
 \begin{equation}\label{Eq:matrix-entry}
   \delta_{u,v} \Bigl( V(u)\ +\ \sum_{(u,w)\in\calE(\Gamma)} e_{(u,w)} \Bigr)\ -\
   \sum_{(u,a+v)\in\calE(\Gamma)} e_{(u,a+v)} z^a\ .
 \end{equation}
%

\begin{Example}\label{Ex:Graphene_operator}
  Let  $\Gamma$ be the hexagonal lattice from Figure~\ref{F:first_graphs}.
  Figure~\ref{F:localGraphene} shows a labeling in a neighborhood of its fundamental domain.
  \begin{figure}[htb]
    \centering
   \begin{picture}(105,92)(0,-1)
     \put(0,0){\includegraphics{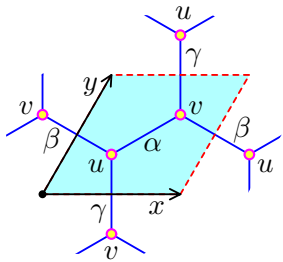}}
     \put(31,31){\small$u$}      \put(95,31){\small$u$}     \put(64,89){\small$u$}
     \put(69,53){\small$v$}      \put( 5,53){\small$v$}     \put(37,-2){\small$v$}
     \put(54,15){\small$x$}
     \put(29,62){\small$y$}
     \put(52,39){\small$\alpha$}
     \put(32,14){\small$\gamma$}   \put(68,73){\small$\gamma$}
     \put(14,40){\small$\beta$}    \put(86,44){\small$\beta$}
   \end{picture}
   \caption{A labeling of the hexagonal lattice.}
   \label{F:localGraphene}
 \end{figure}
  Thus $W=\{u,v\}$ consists of two vertices and there are three (orbits of) edges, with labels $\alpha,\beta,\gamma$.
  Let $(x,y)\in\TT^2$.
  The operator $L$ is
 \begin{align*}
   L(\hat{f})(u)\  &=\ V(u)\hat{f}(u)\ +\ \alpha(\hat{f}(u)-\hat{f}(v))\ +\
                          \beta(\hat{f}(u)-x^{-1}\hat{f}(v))+\gamma(\hat{f}(u)-y^{-1}\hat{f}(v))\ ,\\
   L(\hat{f})(v)\  &=\ V(v)\hat{f}(v)\ +\ \alpha(\hat{f}(v)-\hat{f}(u))\ +\
                         \beta(\hat{f}(v)-x\hat{f}(u))\ \ \ \  + \gamma(\hat{f}(v)-y\hat{f}(u))\ .
 \end{align*}
 Collecting coefficients of $\hat{f}(u),\hat{f}(v)$, we represent $L$ by the $2\times 2$-matrix,
 \begin{equation}\label{Eq:Graphene_Matrix}
    L\ =\ \left( \begin{matrix}
          V(u)+\alpha+\beta+\gamma    &-\alpha-\beta x^{-1}-\gamma y^{-1}\\
         -\alpha-\beta x-\gamma y  & V(v)+\alpha+\beta+\gamma
              \end{matrix}\right)\ ,
 \end{equation}
 whose entries are Laurent polynomials in $x,y$.
 Notice that the support $\calA(\Gamma)$ of $\Gamma$ equals the set of exponents of monomials which appear in $L$.
 Observe that for $(x,y)\in\TT^2$, $L^T=\overline{L}$, so that $L$ is Hermitian, showing again that the operator $L$
 is self-adjoint.\hfill$\diamond$
\end{Example}

What we saw in Example~\ref{Ex:Graphene_operator} holds in general.
In the standard basis for $\CC^W$, $L=L_c$ is multiplication by a $|W|\times|W|$-matrix \defcolor{$L(z)=L_c(z)$} with each 
entry~\eqref{Eq:matrix-entry} a finite sum of monomials with exponents from $\calA(\Gamma)$
(a \demph{Laurent polynomial with support $\calA(\Gamma)$}).
Note that $(u,a{+}v)\in\calE(\Gamma)$ if and only if $(-a{+}u,v)\in \calE(\Gamma)$, these edges have the same label, and for
$z\in\TT^d$, $\overline{z^a}=z^{-a}$.
Thus for $z\in\TT^d$, the matrix is Hermitian, as $L(z)^T=L(z^{-1})=\overline{L(z)}$.

\subsection{Critical points of the Bloch variety}
As $L(z)$ is Hermitian for $z\in\TT^d$, its spectrum is real and consists of its $|W|$ eigenvalues
 \begin{equation}\label{Eq:bandFunctions}
   \lambda_1(z)\ \leq\ \lambda_2(z)\ \leq\ \dotsb\ \leq\ \lambda_{|W|}(z)\,.
 \end{equation}
 These eigenvalues vary continuously with $z\in\TT^d$, and $\lambda_j(z)$ is called the $j$th
 \demph{spectral band function}, $\lambda_j\colon\TT^d\to\RR$.
 Its image is an interval in $\RR$, called the \demph{$j$th spectral band}.
The eigenvalues~\eqref{Eq:bandFunctions} are the roots of the characteristic polynomial
 \begin{equation}\label{Eq:charPoly}
   \defcolor{D(z,\lambda)}\ =\ D_c(z,\lambda)\  \vcentcolon=\ \det(L_c(z)\ -\ \lambda I)\,,
 \end{equation}
 which we call the \demph{dispersion function}.
 Its vanishing defines a hypersurface
 \begin{equation}
   \defcolor{\Var(D_c(z,\lambda))} = \{(z,\lambda) \in \TT^d\times\RR \mid D(z,\lambda)=0 \}\,,
 \end{equation}
 called the \demph{Bloch variety} of the operator $L$\footnote{This is also called the dispersion relation in the literature. We use
 the term Bloch variety as it is an algebraic variety and our perspective is to use methods from algebraic geometry in spectral theory.}.
 The Bloch variety is the union of $|W|$ branches with the $j$th branch equal to the graph of the $j$th spectral band function.
The image of the Bloch variety under the projection to $\RR$ is the spectrum \defcolor{$\sigma(L)$} of the operator $L$.
This projection is a function \defcolor{$\lambda$} on the Bloch variety.
Identifying the $j$th branch/graph with $\TT^d$, the restriction of $\lambda$ to that branch gives the corresponding spectral band
  function $\lambda_j$.

Figure~\ref{F:smooth_graphene} shows this for the operator $L$ on the hexagonal lattice with edge weights $6,3,2$ and zero potential
$V$---for this we unfurl $\TT^2$, representing it by 
$[-\frac{\pi}{2},\frac{3\pi}{2}]^2\subset\RR^2$, which is a fundamental domain in its universal cover.
\begin{figure}[htb]
  \centering
  \begin{picture}(178,138)(-39,0)
     \put(0,0){\includegraphics[height=140.4pt]{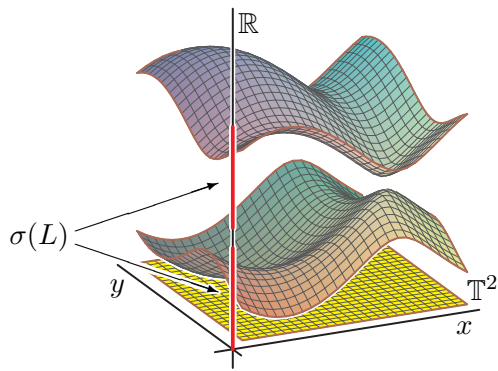}}
     \put(130,15){\small$x$} \put(-1,31){\small$y$}
     \put(47.5,129){\small$\RR$}    \put(134,28){\small$\TT^2$}
     \put(-39,49){\small$\sigma(L)$}
     \thicklines
       \put(-14,50){{\color{white}\line(3,-1){52}}}
       \put(-14,49.5){{\color{white}\line(3,-1){52}}}
       \put(-14,49){{\color{white}\line(3,-1){52}}}
       \put(-14,48.5){{\color{white}\line(3,-1){52}}}
       \put(-14,48){{\color{white}\line(3,-1){52}}}
     \thinlines
       \put(40,31.){\color{white}{\circle*{1.5}}}
       \put(39,31.3333){\color{white}{\circle*{2}}}
       \put(38.5,31.5){\color{white}{\circle*{2}}}
       \put(37,32){\color{white}{\circle*{3.5}}}
       \put(-14,54){\vector(3,1){54}}    \put(-14,49){\vector(3,-1){54}}
   \end{picture}
   \caption{A Bloch variety and spectral bands for the hexagonal lattice.}
   \label{F:smooth_graphene}
\end{figure}
(That is, by quasimomenta in $[-\frac{\pi}{2},\frac{3\pi}{2}]^2$.)
It has two branches with each the graph of the corresponding spectral band function.
An endpoint of a spectral band (\demph{spectral edge}) is the image of an extremum of some band function $\lambda_j(z)$.
For the hexagonal lattice at these parameters, each band function has two nondegenerate extrema, and these give the four spectral edges.
These are also local extrema of the function $\lambda$ on the Bloch variety.

The \demph{spectral edges conjecture}~\cite[Conj.\ 5.25]{KuchBAMS} for a periodic graph $\Gamma$ asserts that for generic values of the
parameters 
$(e,V)$, each spectral edge is attained by a single band, the extrema on the Bloch variety are isolated, and all extrema are nondegenerate
(the spectral band function $\lambda_j$ has a full rank Hessian matrix).
Here, generic means that there is a nonconstant polynomial $p(e,V)$ in the parameters such that when $p(e,V)\neq 0$, these desired
 properties hold.

The entries in the matrix $L(z)$ and the function~\eqref{Eq:charPoly} defining the Bloch variety are all (Laurent) polynomials.
In this setting it is natural to allow complex parameters, $e\colon\calE(\Gamma)\to\CC$, $V\colon\calV(\Gamma)\to\CC$ and variables
$z\in(\CC^\times)^d$, $\lambda\in\CC$.
With complex parameters and variables, $L_c(z)$ is no longer Hermitian, but it does satisfy $L_c(z)^T=L_c(z^{-1})$ and the Bloch variety is
the complex algebraic hypersurface $\Var(D_{c}(z,\lambda))$ in $(\CC^\times)^d\times\CC$ defined by the vanishing of the 
dispersion function $D_{c}(z,\lambda)$ of $L_c(z)$~\eqref{Eq:charPoly}.

In passing to the complex Bloch variety we may no longer distinguish branches $\lambda_j(z)$ of $\lambda$.
At a smooth point $(z_0,\lambda_0)$ whose projection $z$ to $(\CC^\times)^d$ is regular (in that
$\frac{\partial D}{\partial\lambda}(z_0,\lambda_0)\neq 0$), there is a locally defined function $f$ of
$z$ with $\lambda_0=f(z_0)$ and $D(z,f(z))=0$ on its domain, but this is not necessarily a global function of $z$.
Consequently, we will consider the projection to the last coordinate to be a function $\lambda$ on the Bloch variety, and then study its
differential geometry, including its critical points.

Nondegeneracy of spectral edges is implied by the stronger condition that all critical points of the function $\lambda$ on the complex
Bloch variety are nondegenerate.
Understanding the critical points of $\lambda$ is a first step.
Our aim is to bound the number of (isolated) critical points of $\lambda$ on the Bloch variety of a given operator $L$, give criteria for when
the bound is attained, prove that it is attained for generic operators on a class of graphs, and finally to use these results to prove
the spectral edges conjecture for $2^{19}+2$ graphs.
We treat these in the following four sections.

\section{Bounding the number of critical points}\label{S:A}

We first recast extrema of spectral band functions in terms of constrained optimization.
The complex Bloch variety is the hypersurface $\Var(D(z,\lambda))$ in $(\CC^\times)^d\times\CC$ defined by the vanishing of the dispersion
function $D(z,\lambda)$.
\defcolor{Critical points} of the function $\lambda$ on the Bloch variety
are points of the Bloch variety where the gradients in $(\CC^\times)^d\times\CC$ of $\lambda$ and $D(z,\lambda)$ are linearly dependent.
  That is, a critical point is a point $(z,\lambda)\in(\CC^\times)^d\times\CC$ with $D(z,\lambda)=0$ such that either the
  gradient $\nabla D(z,\lambda)$ vanishes or we have
  $\frac{\partial D}{\partial z_i}(z,\lambda)=0$ for $i=1,\dotsc, d$ and $\frac{\partial D}{\partial\lambda}(z,\lambda)\neq0$ (as
  $\nabla\lambda=(0,\dotsc,0,1)$).
  In either case, we have
\[
  D(z,\lambda)\ =\ 0
    \qquad\mbox{and}\qquad
  \frac{\partial D}{\partial z_i}\ =\ 0 \qquad\mbox{for } i=1,\dots,d\,.
\]
Since $z_i\neq 0$, we obtain the equivalent system
 \begin{equation}\label{Eq:CPE}
   D(z,\lambda)\ =\
   z_1\frac{\partial D}{\partial z_1}\ =\ \ \dotsb\    \ =\
   z_d\frac{\partial D}{\partial z_d}\ =\ 0 \ ,
 \end{equation}
which we call the \demph{critical point equations}.

\begin{Proposition}\label{P:CPE}
  A point $(z,\lambda)\in(\CC^\times)^d\times\CC$ is a critical point of the function $\lambda$ on the Bloch variety
  $\Var(D(z,\lambda))$ if and only if~\eqref{Eq:CPE} holds.
\end{Proposition}
\begin{proof}
    We already showed that at a critical point of $\lambda$, the equations~\eqref{Eq:CPE} hold.
    Suppose now that $(z,\lambda)\in(\CC^\times)^d\times\CC$ is a solution to~\eqref{Eq:CPE}.
    As $D(z,\lambda)=0$, the point lies on the Bloch variety.
    As $z\in(\CC^\times)^d$, 
    no coordinate $z_i$ vanishes, which implies that $\frac{\partial D}{\partial z_i}(z,\lambda)=0$ for $i=1,\dotsc,d$.
    Thus the gradients $\nabla \lambda$ and $\nabla D$ are linearly dependent at $(z,\lambda)$, showing that it is a critical point.
\end{proof}
\begin{Remark}
 A point $(z_0,\lambda_0)\in\TT^d\times\RR$ such that $\lambda_0=\lambda_j(z_0)$ is an extreme value of the spectral band function
 $\lambda_j$ is also a critical point of the Bloch variety.
 Indeed, either the gradient $\nabla D$ vanishes at $(z_0,\lambda_0)$ or it does not vanish.
 If $\nabla D(z_0,\lambda_0)=\bfz$, then $(z_0,\lambda_0)$ is a critical point.
If $\nabla D(z_0,\lambda_0)\neq 0$, then the Bloch variety is smooth at $(z_0,\lambda_0)$ and thus is a smooth point of the graph of $\lambda_j$.
As $\lambda_0=\lambda_j(z_0)$ is an extreme value of $\lambda_j$, the tangent plane is horizontal at $(z_0,\lambda_0)$.
This implies that $\lambda_j$ is differentiable (by the implicit function theorem) and that
$\frac{\partial \lambda_j}{\partial z_i}(z_0,\lambda_0)=0$ for $i=1,\dotsc,d$.
Thus the gradients of $\lambda$ and $D$ at $(z_0,\lambda_0)$ are linearly dependent, showing that it is a
critical point.\hfill$\diamond$ 
\end{Remark}

B\'ezout's Theorem~\cite[Sect.\ 4.2.1]{Shaf13} gives an upper bound on the number of isolated critical points:
We may multiply each Laurent polynomial in~\eqref{Eq:CPE} by a monomial to clear denominators and obtain ordinary polynomials.
The product of their degrees is an upper bound for the number of the common zeroes that are isolated in the complex domain.
Polyhedral bounds that exploit the structure of the Laurent polynomials are typically much smaller.
Sources for these are~\cite[Ch.\ 7]{CLOII}, \cite[Ch.\ 5]{GKZ}, and~\cite[Ch.\ 3]{IHP}.
These results bound the number of isolated common zeroes, counted with multiplicities. An isolated common zero $z_0$ of polynomials
$f_1,\dots,f_{d+1}$ on $(\CC^\times)^d\times\CC$ has multiplicity 1 exactly when the gradient of $f_1,\dots,f_{d+1}$ spans the cotangent
space at $z_0$; otherwise its multiplicity exceeds 1 (see~\cite[Ch.\ 4\ Def.\ 2.1]{CLOII} and~\cite[Ch.\ 8.7 \ Def.\ 8]{CLO}). 

Let \defcolor{$\CC[z^{\pm},\lambda]$} be the ring of Laurent polynomials in $z_1,\dotsc,z_d,\lambda$ where $\lambda$ occurs with
only nonnegative exponents.
Note that $D(z,\lambda)\in  \CC[z^{\pm},\lambda]$.
The \demph{support} $\defcolor{\calA(\psi)}\subset\ZZ^d\times\NN$ of a polynomial $\psi\in\CC[z^{\pm},\lambda]$ is the set of exponents of
monomials in $\psi$. 
The \defcolor{Newton polytope} $\defcolor{\calN(\psi)}\vcentcolon=\conv(\calA(\psi))$ of $\psi$ is the convex hull of its support.
Write \defcolor{$\vol(\calN(\psi))$} for the $(d{+}1)$-dimensional Euclidean volume of the Newton polytope of $\psi$.

\begin{Example}\label{Ex:GraphenePolytope}
  We continue the example of the hexagonal lattice.
  Writing \defcolor{$\ell$} for $\alpha{+}\beta{+}\gamma {-}\lambda$, the dispersion function \defcolor{$D(x,y;\lambda)$} of the
  matrix~\eqref{Eq:Graphene_Matrix} is 
  \[
  (V(u)+\ell)(V(v)+\ell)\ -\ 
  (-\alpha-\beta x^{-1}-\gamma y^{-1})(-\alpha-\beta x-\gamma y)\,.
  \]
  In Figure~\ref{F:GraphenePolytope} the monomials in $D(x,y;\lambda)$ label the columns of a $3\times 9$ array which are their
  exponent vectors.
  Figure~\ref{F:GraphenePolytope} also shows its Newton polytope, which has volume 2.\hfill$\diamond$
  \begin{figure}[htb]
  \centering
  \begin{tabular}{ccccccccc}
    $x$&$xy^{-1}$&$y^{-1}$&$x^{-1}$&$x^{-1}y$&$y$&$1$&$\lambda$&$\lambda^2$\\
    $1$& $1$    &  $0$  & $-1$  &  $-1$  &$0$&$0$&$0$&     $0$   \\
    $0$& $-1$   & $-1$  &  $0$  &   $1$  &$1$&$0$&$0$&     $0$   \\
    $0$&  $0$   &  $0$  &  $0$  &   $0$  &$0$&$0$&$1$&     $2$
  \end{tabular}  
  \qquad\qquad
  \raisebox{-45pt}{\begin{picture}(85,100)(-5,0)
        \put(-5,-5){\includegraphics{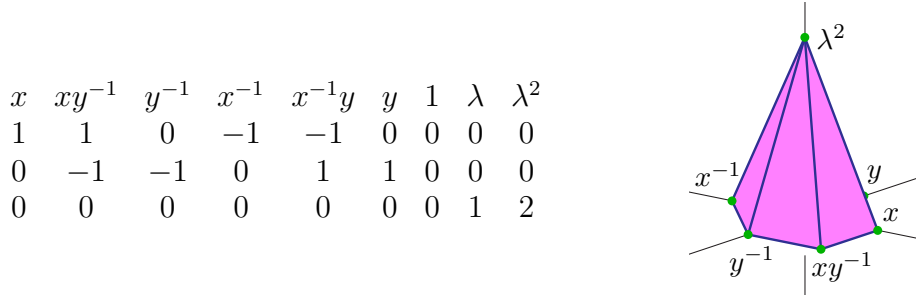}}
            \put(43,88){\small$\lambda^2$}
        \put(68,23){\small$x$}       \put(41,2){\small$xy^{-1}$}
        \put(10,4){\small$y^{-1}$}   \put(-3,36){\small$x^{-1}$}
        \put(62,39){\small$y$}
  \end{picture}}
  
  \caption{Support  and Newton polytope of the hexagonal lattice operator.}
  \label{F:GraphenePolytope}
\end{figure}

\end{Example}

\begin{Theorem}\label{Th:Bound}
  For a polynomial $\psi\in\CC[z^{\pm},\lambda]$, the critical point equations for $\psi$
  \begin{equation}\label{Eq:psi_system}
    \psi(z,\lambda)\ =\
    z_1 \frac{\partial\psi}{\partial z_1}\ =\ \dotsb\ =\ z_d \frac{\partial\psi}{\partial z_d}\ =\ 0
  \end{equation}
  have at most $(d{+}1)!\vol(\calN(\psi))$ isolated solutions in $(\CC^\times)^d\times\CC$, counted with multiplicity.
  When the bound is attained, all solutions are isolated.
\end{Theorem}

We prove this at the end of the section.

As the Bloch variety is defined by the dispersion function $D(z,\lambda)=\det(L(z)-\lambda I)$,
we deduce the following from Theorem~\ref{Th:Bound}.
  
\begin{Corollary}\label{C:Bound}
  The number of isolated critical points of the function $\lambda$ on the Bloch variety for an operator $L$ on a discrete periodic graph 
  is at most $(d{+}1)!\vol(\calN(D))$.  
\end{Corollary}

Theorem A follows from this and Lemma~\ref{L:genericNP}, which asserts that
  \[ \calN(D)\ \subset\ |W|(\conv(\calA(\Gamma)\cup\{\bfe\})\,,\]
  where $\bfe=(0,\dotsc,0,1)$.
  This containment implies the inequality
  \[
  (d{+}1)!\vol(\calN(D))\ \leq
  \ (d{+}1)!|W|^{d+1}\vol(\conv(\calA(\Gamma)\cup\{\bfe\}))\ =\ d!\,|W|^{d+1}\vol(\calA(\Gamma))\,.
  \]

We prove Theorem~\ref{Th:Bound} and Corollary~\ref{C:Bound} after developing some preliminary results.

\subsection{A little algebraic geometry}\label{Sec:AG}
For more from algebraic geometry, see~\cite{CLO,Shaf13}.
An (affine) \demph{variety} is the set of common zeroes of some polynomials $f_1,\dotsc,f_r\in\CC[x_1,\dotsc,x_n]$,
 \[
   \defcolor{\Var(f_1,\dotsc,f_r)}\ \vcentcolon=\  \{x\in\CC^n\mid f_1(x)=\dotsb=f_r(x)=0\}\,.
 \]
We also call this the set of \demph{solutions} to the system $f_1=\dotsb=f_r=0$.   
We may replace any factor $\CC$ in $\CC^n$ by $\CC^\times$, and then allow the corresponding variable to have negative exponents.
The complement of a variety $X$ is a (Zariski) open set.
This defines the \demph{Zariski topology} in which varieties are the closed sets.
A variety is irreducible if it is not the union of two proper subvarieties.
For an irreducible variety, any nonempty open set is dense (even in the classical topology) and any nonempty classically open set is dense in
the Zariski topology. 
Maps $f\colon X\to Y\subset\CC^m$ of varieties are given by $m$ polynomials on $X$ and the image $f(X)$ contains an
open subset of its closure.

Suppose that $X=\Var(f_1,\dotsc,f_r)$.
The \demph{smooth (nonsingular)} locus of $X$ is the open subset of points  of $X$ where the
Jacobian of $f_1\dotsc,f_r$ has maximal rank on $X$.
Let $f$ be a single polynomial.
A point $x$ is a smooth point on the hypersurface $\Var(f)$ defined by $f$ if $f(x)=0$, so that $x\in\Var(f)$ and if the
gradient $\nabla f(x)=(\frac{\partial f}{\partial x_1}(x),\dotsc,\frac{\partial f}{\partial x_n}(x))$ is nonzero,
so that some partial derivative of $f$ does not vanish at $x$.
The point $x\in\Var(f)$ is \demph{singular} if all partial derivatives of $f$ vanish at $x$.
The kernel of the Jacobian at $x\in X=\Var(f_1,\dotsc,f_r)$ is the (Zariski) tangent space at $x$.
The dimension of an irreducible variety is the dimension of a tangent space at any smooth point.
An isolated point $x$ of $X$ has multiplicity one exactly when it is nonsingular.

\begin{Remark}
    Our definition of smooth and singular points of a variety depends upon its defining polynomials.
  For example, the variety defined by $(z-\lambda)^2$ is singular at every point.
  This {\sl scheme-theoretic} notion of singularity is essential to our arguments in
  Sections~\ref{S:B} and~\ref{S:C}, and is standard in algebraic geometry.\hfill$\diamond$
\end{Remark}

If $X$ is irreducible, then any proper subvariety has smaller dimension.
If $f\colon X\to Y$ is a map of varieties with $f(X)$
dense in $Y$, then there is an open subset $U$ of $Y$ such that if $y\in U$, then $\dim f^{-1}(y) + \dim Y = \dim X$.
We also have \demph{Bertini's Theorem}: if $X$ is smooth, then $U$ may be chosen so that for every $y\in U$, $f^{-1}(y)$ is smooth.

Projective space \defcolor{$\PP(\CC^n)$} is the set of one-dimensional linear subspaces (lines) of $\CC^n$ and is compact.
It has dimension $n{-}1$ and subvarieties are given by homogeneous polynomials.
The set \defcolor{$U_0$} of lines spanned by vectors whose initial coordinate is nonzero is isomorphic to $\CC^{n-1}$ under
$v\mapsto\mbox{span}(1,v)$ and $\PP(\CC^n)$ is a compactification of $U_0\simeq\CC^{n-1}$.

\subsection{Polyhedral bounds}\label{Sec:bounds}

The expression $(d{+}1)!\vol(\calN(\psi))$ of Theorem~\ref{Th:Bound} is the \demph{normalized volume} of $\calN(\psi)$.
This is Kushnirenko's bound~\cite[Ch.\ 6, Thm.\ 2.2]{GKZ} for the number of isolated solutions in
$(\CC^\times)^{d+1}$ to a system of $d{+}1$ polynomial equations, all with Newton polytope $\calN(\psi)$.
To prove Theorem~\ref{Th:Bound}, we first explain why Kushnirenko's bound applies to the system~\eqref{Eq:psi_system}, and then why it
bounds the number of isolated solutions on the larger space $(\CC^\times)^d\times\CC$.

For a monomial $z^a\lambda^j$ in $\CC[z^{\pm},\lambda]$, $a\in\ZZ^d$ and $j\in\NN$.
For each $i=1,\dotsc,d$, this monomial is an eigenvector for the operator $z_i\frac{\partial}{\partial z_i}$ with eigenvalue $a_i$.
Thus $\calA(z_i\frac{\partial}{\partial z_i}\psi)\subset\calA(\psi)$, giving the inclusion
$\calN(z_i\frac{\partial}{\partial z_i}\psi)\subset\calN(\psi)$.
A refined version of Kushnirenko's Theorem in which the polynomials may have different Newton polytopes is Bernstein's
theorem~\cite[Sect.\ 7.5]{CLOII}, which is in terms of a quantity called mixed volume, whose properties are developed
in~\cite[Ch.\ IV]{Ewald}. 
The mixed volume of polytopes is monotone under inclusion of polytopes and it equals the normalized volume when all polytopes coincide.
It follows that the theorems of Bernstein and Kushnirenko together give the bound of $(d{+}1)!\vol(\calN(\psi))$ for the number of
isolated solutions to the system~\eqref{Eq:psi_system} in $(\CC^\times)^{d+1}$.
To extend this to solutions in the larger space $(\CC^\times)^d\times\CC$, we develop some theory of projective toric varieties.

\subsection{Projective toric varieties}\label{Sec:Toric}
For Kushnirenko's Theorem and our extension, we replace the nonlinear equations~\eqref{Eq:psi_system} on
$(\CC^\times)^d\times\CC$ by linear equations on a projective variety.
We follow the discussion of~\cite[Ch.\ 3]{IHP}.
Let $\defcolor{f}\in\CC[z^{\pm},\lambda]$ be a polynomial with support $\calA=\calA(f)$.
To  simplify the presentation, we will at times assume that the origin $\bfz$ lies in $\calA$.
The results hold without this assumption, as explained in ~\cite[Ch.\ 3]{IHP}.

Writing \defcolor{$\CC^\calA$} for the vector space with basis indexed by elements of $\calA$, consider the map
\begin{eqnarray*}
  \varphi_\calA\ \colon\ (\CC^\times)^d\times\CC &\longrightarrow& \CC^\calA\\
  (z,\lambda)&\longmapsto & (z^a\lambda^j \mid (a,j)\in\calA)\ .
\end{eqnarray*}
This map linearizes nonlinear polynomials.
Indeed, write $f$ as a sum of monomials,
\[
f\ =\ \sum_{(a,j)\in\calA} c_{(a,j)}z^a\lambda^j \ .
\]
If $\{x_{(a,j)}\mid (a,j)\in\calA\}$ are variables (coordinate functions) on $\CC^\calA$, then 
 \begin{equation}\label{Eq:LinearForm}
\defcolor{\Lambda_f}\ \vcentcolon=\ \sum_{(a,j)\in\calA} c_{(a,j)} x_{(a,j)}
 \end{equation}
is a linear form on $\CC^\calA$, and we have $f(z,\lambda)=\Lambda_f(\varphi_\calA(z,\lambda))=\vcentcolon\defcolor{\varphi^*_\calA(\Lambda_f)}$.

Since $\bfz\in\calA$, the corresponding coordinate $x_{\bfz}$ of $\varphi_\calA$ is 1 and so the image of $\varphi_\calA$ lies in the
principal affine open subset $U_{\bfz}$ of the projective space $\defcolor{\PP^\calA}\vcentcolon=\PP(\CC^\calA)=\PP^{|\calA|-1}$.
This is the subset of $\PP^\calA$ where $x_{\bfz}\neq 0$ and it is isomorphic to the affine space $\CC^{|\calA|-1}$.
We define \defcolor{$X_\calA$} to be the closure of the image $\varphi_\calA((\CC^\times)^{d+1})$ in the projective space $\PP^\calA$,
which is a projective toric variety.
Because the map $\varphi_{\calA}$ is continuous on $(\CC^\times)^d\times\CC$, $X_\calA$ is also the closure of
the image $\varphi_\calA((\CC^\times)^d\times\CC)$.

The map $\varphi_\calA$ is not necessarily injective; we describe its fibers.
Let $\ZZ\calA\subset\ZZ^{d+1}$ be sublattice generated by all differences $\alpha{-}\beta$ for $\alpha,\beta\in\calA$.
When $\bfz\in\calA$ this is the sublattice generated by $\calA$, and it has full rank $d+1$ if and only if $\conv(\calA)$ has full dimension
$d+1$. 
Let \defcolor{$G_\calA$} be $\mbox{Hom}(\ZZ^{d+1}/\ZZ\calA,\CC^\times)\subset(\CC^\times)^{d+1}$, which acts on $(\CC^\times)^d\times\CC$.
The fibers of $\varphi_\calA$ are exactly the orbits of $G_\calA$ on $(\CC^\times)^d\times\CC$.
If $\conv(\calA)$ does not have full dimension, then $G_\calA$ has positive dimension as do all fibers of $\varphi_\calA$, otherwise
$G_\calA$ is a finite group and $\varphi_\calA$ has finite fibers.
On the torus $(\CC^\times)^{d+1}$, $G_\calA$ acts freely and $\varphi_\calA((\CC^\times)^{d+1})$ is identified with 
$(\CC^\times)^{d+1}/G_\calA$.
To describe the fibers of $\varphi_\calA$ on $(\CC^\times)^d\times\{0\} = ((\CC^\times)^d\times\CC)\smallsetminus(\CC^\times)^{d+1}$,
note that $(\CC^\times)^{d+1}$ acts on this through the homomorphism $\pi$ that sends its last ($\lambda$) coordinate to $\{1\}$.
Thus the fibers of $\varphi_\calA$ on $(\CC^\times)^d\times\{0\}$ are exactly the orbits of $\pi(G_\calA)\subset(\CC^\times)^d$.

\begin{Proposition}
  The dimension of $X_\calA$ is the dimension of $\conv(\calA)$.
  The fibers of $\varphi_\calA$ on $(\CC^\times)^{d+1}$ are the orbits of $G_\calA$ and its fibers on $(\CC^\times)^d\times\{0\}$ are the
  orbits of $\pi(G_\calA)$. 
\end{Proposition}

We return to the situation of Theorem~\ref{Th:Bound}.
Let $\psi\in\CC[z^{\pm},\lambda]$ be a polynomial with support $\calA$.
As each polynomial in~\eqref{Eq:psi_system} has support a subset of $\calA$, each corresponds to a linear form on $\PP^\calA$ as
in~\eqref{Eq:LinearForm}. 
The corresponding system of linear forms defines a linear subspace \defcolor{$M_\psi$} of $\PP^\calA$.
We have the following proposition (a version of~\cite[Lemma 3.5]{IHP}).

\begin{Proposition}\label{P:bijection}
  The solutions to~\eqref{Eq:psi_system} are the inverse images under $\varphi_\calA$ of 
  points in the linear section $\varphi_\calA((\CC^\times)^d\times\CC)\cap M_\psi$.
  When $\varphi_\calA$ is an injection, it is a bijection between solutions to~\eqref{Eq:psi_system} on $(\CC^\times)^d\times\CC$ and points
  in  $\varphi_\calA((\CC^\times)^d\times\CC)\cap M_\psi$.
\end{Proposition}

\begin{proof}[Proof of Theorem~\ref{Th:Bound}]
 When $\vol(\calN(\psi))=0$, so that $\calN(\psi)$ does not have full dimension $d+1$, then each fiber of $\varphi_{\calA}$ is
 positive-dimensional and so by Proposition~\ref{P:bijection} there 
 are no isolated solutions to~\eqref{Eq:psi_system}.

 Suppose that $\vol(\calN(\psi))>0$.
 Then every fiber of $\varphi_\calA$ is an orbit of the finite group $G_\calA$.
 Over points of $\varphi_\calA((\CC^\times)^{d+1})$, each fiber consists of $|G_\calA|$ points and over
 $\varphi_\calA((\CC^\times)^d\times\{0\}$ each fiber consists of $|\pi(G_\calA)|\leq|G_\calA|$ points.
 As $X_\calA$ is the closure of $\varphi_\calA((\CC^\times)^d\times\CC)$, the number of isolated points in $X_\calA\cap M_\psi$ is at least
 the number of isolated points in $\varphi_\calA((\CC^\times)^d\times\CC)\cap M_\psi$, both counted with multiplicity.
 The degree of the projective variety $X_\calA$ is an upper bound for the number of isolated points in $X_\calA\cap M_\psi$,
 which is explained in~\cite[Ch.\ 3.3]{IHP}.
 There, the product of $|G_\calA|$ and the degree of $X_\calA$ is shown to be $(d{+}1)!\vol(\calN(\psi))$,
 the normalized volume of the Newton polytope of $\psi$.
 This gives the bound of Theorem~\ref{Th:Bound}.
 That all points are isolated when the bound of the degree is attained is Proposition~\ref{P:sharp} in the next section.
\end{proof}


\section{Proof of Theorem B}\label{S:B}

We give conditions for when the upper bound of Corollary~\ref{C:Bound} is attained.
By Proposition~\ref{P:CPE}, the critical points of the function $\lambda$ on the Bloch variety $\Var(D)$ are the solutions in
$(\CC^\times)^d\times\CC$ to the critical point equations~\eqref{Eq:CPE}.
Let $\defcolor{\calA}=\calA(D)$ be the support of the polynomial $D$.
The critical points are $\varphi_\calA^{-1}(X_\calA\cap M_{D})$, where $X_\calA\subset\PP^\calA$ is the closure of
$\varphi_\calA((\CC^\times)^d\times\CC)$ and $M_{D}$ is the subspace of $\PP^\calA$ defined by linear forms corresponding (as
in~\eqref{Eq:LinearForm}) to the polynomials in~\eqref{Eq:CPE}. 
For the bound of Theorem~\ref{Th:Bound} and Corollary~\ref{C:Bound}, note that the number of isolated points of $X_\calA\cap M_{D}$
is at most the product of the degree of $X_\calA$ with the cardinality of a fiber of $\varphi_\calA$, which is
$(d{+}1)!\vol(\calN(D))$.
We establish Theorem B concerning the sharpness of this bound by characterizing when the inequality of
Theorem~\ref{Th:Bound} is strict and then interpreting that for the critical point equations.

\begin{Remark}\label{R:excess}
  Let $X\subset\PP^n$ be a 
  variety of dimension $d$ and $M\subset\PP^n$ a linear subspace of codimension $d$.
  The number of points in $X\cap M$ does not depend on $M$ when the intersection is transverse; it is the \demph{degree} of 
  $X$~\cite[p.\ 234]{Shaf13}.
  When the intersection is not transverse, intersection theory gives a refinement~\cite[Ch.\ 6]{Fulton}.
  For each irreducible component $Z$ of the intersection $X\cap M$, there is a positive integer---the intersection multiplicity along
  $Z$--such that the sum of these multiplicities is the degree of $X$.
  When $Z$ is positive-dimensional this number is the degree of a zero-cycle constructed on $Z$ (it is at least the degree of $Z$) and when
  $Z$ is zero-dimensional 
  (a point), it is the local multiplicity~\cite[Ch.\ 4]{Shaf13}.
  %
  %
  %
 \hfill$\diamond$
\end{Remark}  

A consequence of Remark~\ref{R:excess} is the following.

\begin{Proposition}\label{P:sharp}
  Let $X,M$ be as in Remark~\ref{R:excess}.
  The number (counted with multiplicity) of isolated points of $X\cap M$ is strictly less than the degree of $X$ if and only if 
  the intersection has a positive-dimensional component.
\end{Proposition}

Write $\defcolor{X^\circ_\calA}\vcentcolon=\varphi_\calA((\CC^\times)^d\times\CC)$ for the image of $\varphi_\calA$ and
$\defcolor{\partial X_\calA}\vcentcolon=X_\calA\smallsetminus X^\circ_\calA$, the points of $X_\calA$ added to $X^\circ_\calA$ when taking the closure.
This is the \demph{boundary} of $X_\calA$.
In the Introduction, points of $\partial X_\calA$ were referred to as `lying at infinity'.

\begin{Corollary}\label{C:MeetBoundary}
  For a polynomial $\psi\in\CC[z^{\pm},\lambda]$, the inequality of Theorem~\ref{Th:Bound} is strict if and only if
  $\partial X_\calA\cap M_\psi\neq\emptyset$.
\end{Corollary}
\begin{proof}
  The inequality of Theorem~\ref{Th:Bound} is strict if either of the following hold.
  \begin{enumerate}
   \item $X_\calA\cap M_\psi$ has an isolated point not lying in $X^\circ_\calA$.
   \item $X_\calA\cap M_\psi$ contains a positive-dimensional component $Z$.
  \end{enumerate}
  In (1), $X_\calA\cap M_\psi$ has isolated points in $\partial X_\calA\cap M_\psi$, so the intersection is
  nonempty.
  In (2), $Z$ is a projective variety of dimension at least one.
  The set $X^\circ_\calA$ is an affine variety, and we cannot have $Z\subset X^\circ_\calA$ as the only projective
  varieties that are also subvarieties of an affine variety are points.
  Thus $Z\cap \partial X_\calA\neq\emptyset$, which completes the proof.
\end{proof}

\subsection{Facial systems}\label{Sec:facial}
We return to the general case of a toric variety.
Let $\calA\subset\ZZ^n$ be a finite set of points with corresponding projective toric variety $X_\calA\subset\PP^\calA$.
We have the following description of the points of its boundary,
$X_\calA\smallsetminus\varphi_\calA((\CC^\times)^n)$. 

Let $\defcolor{P}\vcentcolon=\conv(\calA)$, the convex hull of $\calA$.
The dot product with a nonzero vector $w\in\RR^n$, $a\mapsto w\cdot a$, defines a linear function on $\RR^n$.
For $w\in\RR^n$, set $\defcolor{h(w)}\vcentcolon=\min\{w\cdot a\mid a\in P\}$.
The set $F=\{p\in P\mid w\cdot p=h(w)\}$ of minimizers is the \demph{face} of $P$ \demph{exposed} by $w$.
We have that $F=\conv(F\cap\calA)$, and may write \defcolor{$\calF$} for $F\cap\calA$.
As $\calA\subset\ZZ^n$, we only need integer vectors $w\in\ZZ^n$ to expose all faces of $P$.
If $\dim F = \dim P-1$, then $F$ is a \demph{facet}.

For each face $F$ of $P$, there is a corresponding coordinate subspace \defcolor{$\PP^\calF$} of $\PP^\calA$---this
is the set of points $z=[z_a\mid a\in\calA]\in\PP^\calA$ such that $a\not\in F$ implies that $z_a=0$.
The image of the map $\varphi_{\calF}\colon(\CC^\times)^n\to\PP^{\calF}\subset\PP^\calA$ has closure the toric variety $X_{\calF}$.
Its dimension is equal to the dimension of the face $F$.
Write \defcolor{$X_{\calF}^\circ$} for the image of $\varphi_{\calF}$.
This description and the following proposition is essentially~\cite[Prop.\ 5.1.9]{GKZ}.

\begin{Proposition}
  The boundary of the toric variety $X_\calA$ is the disjoint union of the sets $X_{\calF}^\circ$ for all the proper faces
  $F$ of $\conv(\calA)$.
\end{Proposition}

Let $f=\sum_{a\in\calA} c_a x^a$ be a polynomial with support $\calA$.
We observed that if $\Lambda$ is the corresponding linear form~\eqref{Eq:LinearForm} on $\PP^\calA$, then the variety
$\Var(f)\subset(\CC^\times)^n$ of $f$ is the pullback along $\varphi_\calA$ of $X^\circ_\calA\cap M$, where
$\defcolor{M}\vcentcolon=\Var(\Lambda)$ is the hyperplane defined by $\Lambda$.
Let $F$ be a proper face of $P$.
Then $X^\circ_{\calF}\cap M$ pulls back along $\varphi_{\calF}$ to the variety of
\[
\varphi^{-1}_{\calF}(\Lambda)\ =\ \sum_{a\in F} c_a x^a
\]
in $(\CC^\times)^n$.
This sum of the terms of $f$ whose exponents lie in $F$ is a \demph{facial form} of $f$ and is written \defcolor{$f|_F$}.
Given a system $\Phi\colon f_1=\dotsb=f_n=0$ involving Laurent polynomials with support $\calA$, the system $f_1|_F=\dotsb=f_n|_F=0$
of their facial forms is the \demph{facial system} \defcolor{$\Phi|_F$} of $\Phi$.

\begin{Corollary}\label{C:Sings}
  Let $M$ be the intersection of the hyperplanes given by the polynomials in a system $\Phi$ of Laurent polynomials
  with support $\calA$.
  For each face $F$ of $\conv(\calA)$, the points of $X^\circ_{\calF}\cap M$ pull back under $\varphi_{\calF}$ to the solutions of the facial
  system $\Phi|_F$. 

  If no facial system $\Phi|_F$ has a solution, then the number of solutions to $\Phi=0$ on $(\CC^\times)^n$ is
  $n!\vol(\conv(\calA))$.
\end{Corollary}
\begin{proof}
 The first statement follows from the observation about a single polynomial $f$ and its facial form $f|_F$, and the second is a consequence
 of a version of Corollary~\ref{C:MeetBoundary} for $X_\calA\smallsetminus\varphi_\calA((\CC^\times)^n)$.  
\end{proof}

The second statement is essentially~\cite[Thm.\ B]{Bernstein} and is also explained in~\cite[Sect.\ 3.4]{IHP}.

\subsection{Facial systems of the critical point equations}\label{S:facial}
We prove Theorem B from the Introduction by interpreting the facial systems of the critical point equations.
It is useful to introduce the following notion.
A polynomial $f(x)$ in $x\in(\CC^\times)^n$ is \demph{quasi-homogeneous} with \demph{quasi-homogeneity}
$w\in\ZZ^n$ if there is a number $0\neq w_f$ such that
\[
     a\ \in\ \calA(f)\ \Longrightarrow\ w\cdot a\ =\ w_f\,.
\]
Equivalently, $f$ is quasi-homogeneous if its support $\calA(f)$ lies on a hyperplane not containing the origin.
The quasi-homogeneities of $f$ are those $w\in\ZZ^n$ whose dot product is constant on $\calA(f)$.
For $t\in \CC^\times$ and $w\in\ZZ^n$, let $\defcolor{t^w}\vcentcolon=(t^{w_1},\dotsc,t^{w_n})\in(\CC^\times)^n$.

\begin{Lemma}\label{L:quasi-homogeneous}
  Suppose that $f$ has a quasi-homogeneity $w\in\ZZ^n$.
  Then
  \begin{enumerate}
    \item For $t\in\CC^\times$ and $x\in(\CC^\times)^n$, we have $f(t^w\cdot x)=t^{w_f} f(x)$.
    \item We have
      \[
      w_f \; f\ =\ \sum_{i=1}^n w_i\,x_i \frac{\partial f}{\partial x_i}\,.
      \]
   \end{enumerate}    
\end{Lemma}
\begin{proof}
  Note that for $a\in\ZZ^n$, $(t^w\cdot x)^a = t^{w\cdot a} x^a$.
  The first statement follows.
  For the second, note that $a_ix^a= x_i \frac{\partial}{\partial x_i} x^a$.
\end{proof}

Let $\psi\in\CC[z^{\pm},\lambda]$ have support $\calA\subset\ZZ^d\times\NN$ and Newton polytope
$\defcolor{P}\vcentcolon=\conv(\calA)$.
We will assume that $P$ has dimension $d{+}1$, and also that
$\calA\cap\ZZ^d\times\{0\}$ is a facet of $\calA$, called its \demph{base}.
Let~\eqref{Eq:psi_system} be the critical point equations for $\lambda$ on $\psi$ and $M_\psi\subset\PP^\calA$ the corresponding linear
subspace of codimension $d{+}1$.

Let $\defcolor{\bfz}\vcentcolon= 0^d$ in $\ZZ^d$ and $\defcolor{\bfe}\vcentcolon=(\bfz,1)$.
The base of $\calA$ is exposed by $\bfe$ and it is the support of $\psi(z,0)$.
A main difference between the sparse equations of Section~\ref{Sec:facial} and the critical point
equations~\eqref{Eq:CPE} is that the critical point equations allow solutions with $\lambda=0$, which is the component of the
boundary of the toric variety corresponding to the base of $\calA$.
A face $F$ of $P$ is \demph{vertical} if it contains a vertical line segment, one parallel to $\bfe$.

\begin{Lemma}\label{L:Singularities}
  Suppose that $F$ is a proper face of $P$ that is not the base of $P$ and is not vertical.
  Then the corresponding facial system of the critical point equations has a solution if and only if the hypersurface
  $\Var(\psi|_F)$ defined by
  $\psi|_F$ in $(\CC^\times)^{d+1}$ is singular.
\end{Lemma}
\begin{proof}
  Let $0\neq w\in\ZZ^{d+1}$ be an integer vector that exposes the face $F$.
  As $F$ is not vertical we may assume that $w_{d+1}$ is nonzero.
  As $F$ is not the base, it lies on an affine hyperplane that does not contain
  the origin, so that $\psi|_F$ is quasi-homogeneous with some quasi-homogeneity $w$.
  Write \defcolor{$w_F$} for the constant $w\cdot a$ for $a\in F$.
  By Lemma~\ref{L:quasi-homogeneous} (2), we have 
  \begin{equation}\label{Eq:Euler}
   w_F\,\psi|_F\ =\ \sum_{i=1}^d  w_i\; z_i\frac{\partial\psi|_F}{\partial z_i} \ +\
       w_{d+1}\; \lambda\frac{\partial\psi|_F}{\partial \lambda}\ .
  \end{equation}
  Suppose now that $(z,\lambda)$ is a solution of the restriction of the critical point equations to the face $F$.
  That is, at $(z,\lambda)$,
  \[
     \psi|_\calF\ =\ \left(z_1\frac{\partial\psi}{\partial z_1}\middle)\right|_F\ =\ 
     \dotsb
     \ =\ \left(z_d\frac{\partial\psi}{\partial z_d}\middle)\right|_F\ =\ 0\,.
  \]
  Observe that $(z_i\frac{\partial\psi}{\partial z_i})|_F= z_i\frac{\partial\psi|_F}{\partial z_i}$ 
  (and the same for $\lambda$).
  Since $w_{d+1}\neq 0$, these equations and~\eqref{Eq:Euler} together imply that
  $(\lambda\frac{\partial\psi}{\partial\lambda})|_F=0$,
  which implies that $(z,\lambda)$ is a singular point of the hypersurface $\Var(\psi|_F)$ defined by $\psi|_F$.
\end{proof}

We deduce the following theorem.

\begin{Theorem}\label{Th:SingTheoremA}
  If the Newton polytope $\calN(\psi)$ of $\psi$ has no vertical faces and the restriction of $\psi$ to each
  face that is not the base of $\calN(\psi)$ defines a smooth variety,
  then the critical point equations have exactly $(d{+}1)!\vol(\calN(\calA))$ solutions in $(\CC^\times)^d\times \CC$.
\end{Theorem}

We apply this when $\psi$ is the dispersion function $D(z,\lambda)$.
Recall that the boundary of the variety $X_\calA$
($X_D$) corresponds to all proper faces of its Newton polytope $\calN(D)$, except for its base.
We deduce the following precise version of \textbf{Theorem B}.

\begin{Corollary}\label{C:SingThmB}
  Let $L$ be an operator on a periodic graph and set $D=\det(L(z)-\lambda I)$.
  If $\calN(D)$ has no vertical faces and if for each face $F$ that is not its base, $\Var(D|_F)$ is smooth,
  then the Bloch variety has exactly $(d{+}1)!\vol(\conv(\calA(D)))$ critical points.
\end{Corollary}

\begin{Example}\label{Ex:K4}
  The restriction on vertical faces is necessary.
  General operators on the second graph in Figure~\ref{F:first_graphs} (an abelian cover of $K_4$) have the following
  Newton polytope:
  \[
  \includegraphics[height=100pt]{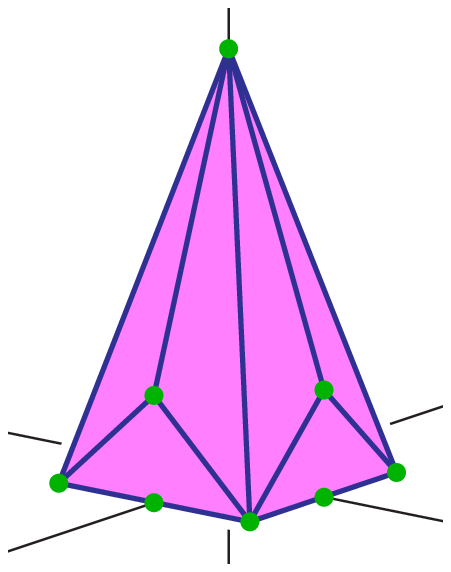}
  \]
  It has base $[-1,1]^2$, apex $(0,0,4)$, and the remaining vertices are at $(\pm 1,0,1)$ and $(0,\pm 1,1)$.
  It has volume $20/3$, so we expect $40=3!\cdot 20/3$ critical points.
  However, there are at most 32 critical points, as direct computation shows that the critical point
  equations have two solutions on each of its four vertical faces.\hfill$\diamond$
\end{Example}

\section{Newton polytopes and dense periodic graphs} \label{S:C}  
The Newton polytope $\calN(D)$ of the dispersion function of an operator on a periodic graph is central to
our results.
In Section~\ref{Sec:41} we associate a polytope $\calN(\Gamma)$ to any periodic graph $\Gamma$ such that $\calN(D)\subset\calN(\Gamma)$ for any
operator on $\Gamma$, and that we have equality for almost all parameter values.
We call $\calN(\Gamma)$ the \demph{Newton polytope} of $\Gamma$.

A periodic graph $\Gamma$ is dense if it has every possible edge, given its support $\calA(\Gamma)$ and fundamental domain $W$.
Every periodic graph is a subgraph of a minimal dense periodic graph.
We identify the Newton polytope of a dense periodic graph and show that when $d=2$ or $3$,
a general operator on $\Gamma$ satisfies Corollary~\ref{C:SingThmB}, which implies Theorem C.

Let $\Gamma$ be a connected $\ZZ^d$-periodic graph with fundamental domain $W$.
Its support $\calA(\Gamma)$ is the finite set of points $a\in\ZZ^d$ such that there is an edge between
$W$ and $a{+}W$.
The integer span of $\calA(\Gamma)$ is $\ZZ^d$, as $\Gamma$ is connected.
The graph $\Gamma$ is \demph{dense} if for every $a\in\calA(\Gamma)$, there is an edge in $\Gamma$ between every pair of vertices in the
union of $W$ and $a{+}W$. 
In particular, the restriction of $\Gamma$ to $W$ is the complete graph on $W$.
The graphs of Figures~\ref{F:denseDimer} and~\ref{F:New-dense} are dense, 
\begin{figure}[htb]
\centering
   \raisebox{-17pt}{\includegraphics[height=140pt]{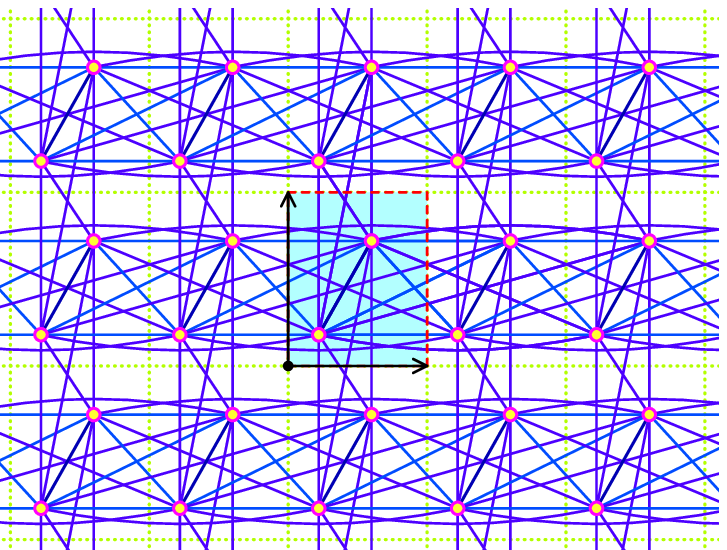}}
    \qquad
   \includegraphics{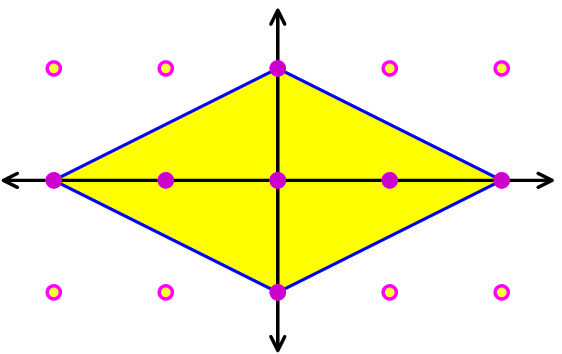}
  \caption{A dense graph $\Gamma$ and its support $\calA(\Gamma)$ with convex hull.}
  \label{F:New-dense}
\end{figure}
while those of Figure~\ref{F:first_graphs} are not dense.

The set of parameters $(e,V)$ for operators on a periodic graph $\Gamma$ is $Y=\CC^E\times\CC^W$, where $E$ is the set of orbits of edges.
We observed that for any $c\in Y$, each entry of $L_c(z)$ has support a subset of $\calA(\Gamma)$.
Consequently, each diagonal entry of $L_c(z)-\lambda I$ has support a subset of $\calA(\Gamma)\cup\{\bfe\}$
and its Newton polytope is a subpolytope of $\defcolor{Q}\vcentcolon=\conv(\calA(\Gamma)\cup\{\bfe\})$.
Let $\defcolor{m}\vcentcolon=|W|$, the number of orbits of vertices.

\begin{Lemma}\label{L:genericNP}
  The Newton polytope $\calN(D_c)$ is a subpolytope of the dilation $mQ$ of $Q$.
\end{Lemma}
\begin{proof}
  The dispersion function $D_c$ is a sum of products of $m$ entries of the $m\times m$ matrix $L_c(z)-\lambda I$.
  Each such product has Newton polytope a subpolytope of $mQ$ as the Newton polytope of a product is the sum of Newton polytopes of the
  factors. 
\end{proof}

Figure~\ref{F:NewtonPolytopes} shows $mQ=2Q$ for the dense graphs of Figures~\ref{F:denseDimer} and~\ref{F:New-dense}.
\begin{figure}[htb]
  \centering
   \includegraphics{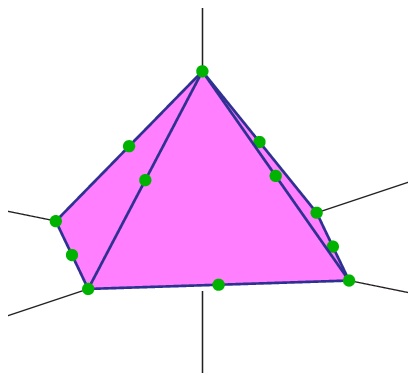}
     \qquad
   \includegraphics{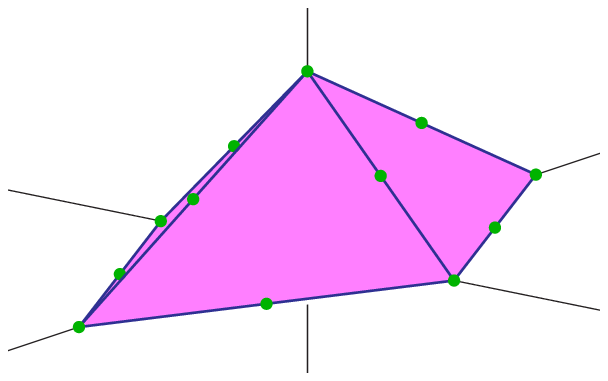}
   \caption{Newton polytopes of dense graphs.}
   \label{F:NewtonPolytopes}
\end{figure}
Observe that $mQ$ is a pyramid with base $m\conv(\calA(\Gamma))$ and apex $m\bfe$, and it has no vertical faces.

\begin{Theorem}\label{Th:denseDense}
  Let $\Gamma$ be a dense $\ZZ^d$-periodic graph.
  There is a nonempty Zariski open subset $U$ of the parameter space $Y$ such that for $c\in U$, the Newton polytope of
  $D_c(z,\lambda)$ is the pyramid $m Q$.
  When $d=2$ or $3$, then we may choose $U$ so that for every $c\in U$ and face $F$ of $m Q$ that is not its base,
  $\Var(D_c|_F)$ is smooth.
\end{Theorem}

Together with Corollary~\ref{C:SingThmB}, this implies Theorem C from the Introduction.
We prove Theorem~\ref{Th:denseDense} in the following two subsections.\medskip

\subsection{The Newton polytope of $\Gamma$}~\label{Sec:41}
For a periodic graph $\Gamma$, the space of parameters $(e,V)$ for operators on $\Gamma$ is $Y=\CC^E\times\CC^W$.
Treating parameters as indeterminates gives the \demph{generic dispersion function} \defcolor{$D(e,V,z,\lambda)$}, which is a
polynomial in $z,\lambda$ whose coefficients are polynomials in the parameters $e,V$.
The \demph{Newton polytope} \defcolor{$\calN(\Gamma)$} of $\Gamma$ is the convex hull of the monomials in $z,\lambda$ that appear in
$D(e,V,z,\lambda)$.

\begin{Lemma}\label{L:NewtonPolytope}
  For  $c\in Y$, $\calN(D_c(z,\lambda))$ is a subpolytope of $\calN(\Gamma)$.
  The set of $c\in Y$ such that $\calN(D_c(z,\lambda))=\calN(\Gamma)$ is a dense open subset $U$.
  When $\Gamma$ is a dense periodic graph, $\calN(\Gamma)=mQ$.
\end{Lemma}
\begin{proof}
  For any $c=(e,V)\in Y$, $D_c(z,\lambda)$ is the evaluation of the generic dispersion function
  $D(e,V,z,\lambda)$ at the point $(e,V)$.
  Thus $\calN(D_c)\subset\calN(\Gamma)$.

  The coefficient \defcolor{$C_{(a,j)}$} of a monomial $z^a\lambda^j$ in $D(e,V,z,\lambda)$ is a polynomial in $(e,V)$.
  For any $c=(e,V)\in Y$, $z^a\lambda^j$ appears in $D_c$ if any only if $C_{(a,j)}(e,V)\neq 0$.
  Thus, we have the equality $\calN(D_c)=\calN(\Gamma)$ of Newton polytopes if and only if $C_{(a,j)}(e,V)\neq 0$
  for every vertex $(a,j)$ of $\calN(\Gamma)$, which defines a dense open subset $U\subset Y$.

  When  $\Gamma$ is dense and no parameter $c$ vanishes, then every diagonal entry of $L_c(z)-\lambda I$ has
  support $\calA(\Gamma)\cup\{\bfe\}$.
  This implies that $\calN(\Gamma)=mQ$.
\end{proof}  

\subsection{Smoothness of the Bloch variety at infinity}
Let $\Gamma$ be a dense periodic graph with $d=2$ or $3$.
Let $U\subset Y$ be the subset of Lemma~\ref{L:NewtonPolytope}.
We show that for each face $F$ of $\calN(\Gamma)$ that is not its base, there is a nonempty open subset $U_F$ of $U$ such that for
$c\in U_F$, the restriction $D_c|_F$ to the monomials in $F$ defines a smooth hypersurface.
Then for parameters $c$ in the intersection of the $U_F$, the operator satisfies the hypotheses of Corollary~\ref{C:SingThmB}, which
proves Theorem~\ref{Th:denseDense} and Theorem C.

Let $F$ be a face of $\calN(\Gamma)$ that is not its base and let $c\in U$.
We may assume that $F$ is not a vertex, for then $D_c|_F$ is a single term and $\Var(D_c|_F)=\emptyset$.
Since $\calN(\Gamma)=mQ$, there is a unique face \defcolor{$G$} of $Q$ such that $F=mG$.
We have that 
\[
    D_c(z,\lambda)|_F\ =\ \det\bigl( (L_c(z)-\lambda I)|_G \bigr)\ ,
\]
where each entry of the matrix $(L_c(z)-\lambda I)|_G$ is the facial form $f|_G$ of the corresponding entry $f$ of $L_c(z)-\lambda I$.

Since $G$ is not the base of $Q$ (and thus does not contain the origin), we make the following observation, which follows from the form of
the operator $L_c$~\eqref{Eq:Laurent_operator}.
If the apex $\bfe=(\bfz,1)$ of $Q$ lies in $G$ and $f$ is a diagonal entry of $(L_c(z)-\lambda I)|_G$, then $f$ contains the term $-\lambda$.
Any other integer point $a\in G$ is nonzero and lies in the support $\calA(\Gamma)$ of $\Gamma$, and the coefficient of $z^a$ in $f$ is
$-e_{(u,a+v)}$, where $f$ is the entry in row $u$ and column $v$.
Consequently, except possibly for terms $-\lambda$, all coefficients of entries in $(L_c(z)-\lambda I)|_G$ are distinct parameters.

Suppose that the fundamental domain is $W=\{v_1,\dotsc,v_m\}$ so that we may index the rows and columns of $L_c(z)$ by $1,\dotsc,m$.
Let $\defcolor{Y'}\subset Y$ be the set of parameters $c$ where
\[
   e_{(v_i,a+v_j)}\ =\ 0\qquad  \mbox{if $a\in G$ \ and \ $j\neq i, i{+}1$.}
\]
(Here, $m{+}1$ is interpreted to be 1.)
For $c\in Y'$, all entries of $L_c(z)|_G$ are zero, except on the diagonal, the first super diagonal, and the lower
left entry.
The same arguments as in the proof of Lemma~\ref{L:NewtonPolytope} show that there exist parameters $c\in Y'$ such that $D_c(z,\lambda)$ has
Newton polytope $\calN(\Gamma)$.
Thus $Y'\cap U\neq\emptyset$, where $U\subset Y$ is the set of Lemma~\ref{L:NewtonPolytope}.

\begin{Theorem}\label{Th:smoothFaces}
 There exists an open subset $U'$ of $Y'$ with $U'\subset U$ such that if $c\in U'$, then $\Var(D_c(z,\lambda)|_F)$ is a smooth
 hypersurface in $(\CC^\times)^{d+1}$.    
\end{Theorem}

Since smoothness of $\Var(D_c(z,\lambda)|_F)$ is an open condition on the space $Y$ of parameters, this will complete the proof of
Theorem~\ref{Th:denseDense}, and thus also of Theorem C.

\begin{proof}
 Let us write \defcolor{$\psi_c(z,\lambda)$} for the facial polynomial $D_c(z,\lambda)|_F$.
 We will show that the set of $c\in Y'$ such that $\Var(\psi_c(z,\lambda))$ is singular is a finite
 union of proper algebraic subvarieties.
 As $c\in Y'$, the only nonzero entries in the matrix $(L_c(z)-\lambda I)|_G$ are its diagonal entries
 $f_1(z,\lambda),\dotsc,f_m(z,\lambda)$ and the entries $g_1(z),\dotsc,g_m(z)$ which are in positions
 $(1,2),\dotsc,(m{-}1,m)$ and $(m,1)$, respectively.
 Thus
 \[
   \psi_c(z,\lambda)\ =\    D_c(z,\lambda)|_F\ =\ \det((L_c(z)-\lambda I)|_G)\ =\
  \prod_{i=1}^m f_i(z,\lambda)\ -\  (-1)^m\prod_{i=1}^m g_i(z)\,.
  \]

 For a polynomial $f$ in the variables $(z,\lambda)$, write \defcolor{$\nabla_\TT$} for the toric gradient operator,
 \[
   \nabla_\TT f \ \vcentcolon=\ 
   \Bigl(z_1\frac{\partial f}{\partial z_1},\dotsc, z_d\frac{\partial f}{\partial z_d}, \lambda\frac{\partial f}{\partial\lambda}\Bigr)\,.
 \]
 Note that 
 \begin{equation}\label{Eq:nablaPsi}
 \nabla_\TT \psi_c\ =\ \sum_{i=1}^m (\nabla_\TT f_i) f_1\dotsb\widehat{f_i}\dotsb f_m
                  \ -\  (-1)^m\sum_{i=1}^m (\nabla_\TT g_i) g_1\dotsb\widehat{g_i}\dotsb g_m\ .
 \end{equation}
 Here $\widehat{f_i}$ indicates that $f_i$ does not appear in the product, and the same for $\widehat{g_i}$.
 
 Let $(z,\lambda)\in\Var(\psi_c)$ be a singular point.
 Then $\psi_c(z,\lambda)=0$ and $\nabla_\TT \psi_c(z,\lambda)=\bfz$.
 There are five cases that depend upon the number of polynomials $f_i,g_j$ vanishing at $(z,\lambda)$.
  \begin{enumerate}[label=(\roman*)]
  \item At least two polynomials $f_p$ and $f_q$ and two polynomials $g_r$ and $g_s$ vanish at $(z,\lambda)$.
    Thus $\psi(z,\lambda)=0$ and by~\eqref{Eq:nablaPsi} this implies that $\nabla_\TT \psi_c(z,\lambda)=\bfz$.
    
    \item At least two polynomials $f_p$ and $f_q$ and exactly one polynomial $g_s$ vanish at $(z,\lambda)$.
      Thus $\psi(z,\lambda)=0$ and by~\eqref{Eq:nablaPsi} if $\nabla_\TT \psi_c(z,\lambda)=\bfz$, then $\nabla_\TT g_s(z,\lambda)=\bfz$.
      
    \item Exactly one polynomial $f_p$ and at least two polynomials $g_r$ and $g_s$ vanish at $(z,\lambda)$.
      Thus $\psi(z,\lambda)=0$ and by~\eqref{Eq:nablaPsi} if $\nabla_\TT \psi_c(z,\lambda)=\bfz$, then $\nabla_\TT f_p(z,\lambda)=\bfz$.
      
    \item Exactly one polynomial $f_p$ and one polynomial $g_r$ vanish at $(z,\lambda)$.
      Thus $\psi(z,\lambda)=0$ and by~\eqref{Eq:nablaPsi} if $\nabla_\TT \psi_c(z,\lambda)=\bfz$, then, after reindexing so that $p=r=1$, we
      have 
      \begin{equation}\label{Eq:iv_consequence}
       \nabla_\TT f_1(z,\lambda)\cdot \prod_{i=2}^m f_i(z,\lambda)\ -\ (-1)^m
       \nabla_\TT g_1(z,\lambda)\cdot \prod_{i=2}^m g_i(z,\lambda)\ =\ \bfz\,.
      \end{equation}

    \item No polynomials $f_i$ or $g_i$ vanish at $(z,\lambda)$.
  \end{enumerate}
  In each case, we will show that the set of parameters $c\in Y'$ such that there exist $(z,\lambda)$ satisfying these
  conditions lies in a proper subvariety of $Y'$.
  Cases (i)---(iv) use arguments based on the dimension of fibers and images of a map and are proven in the rest of this section.
  Case (v) is proven in Section~\ref{S:Case(v)} and it uses Bertini's Theorem.
\end{proof}  

Let us write \defcolor{$X$} for the space $(\CC^\times)^{d+1}$ and \defcolor{$x$} for a point $(z,\lambda)\in X$.
We first derive consequences of some vanishing statements.
For a finite set $\calF\subset\ZZ^{d+1}$, let \defcolor{$\CC^\calF$} be the space of coefficients of polynomials in $x\in X$ with
support $\calF$.
This is the parameter space for polynomials with support $\calF$.

\begin{Lemma}\label{L:linearEquations}
  We have the following.
  
  \begin{enumerate}
   \item For any $x\in X$, $f(x)=0$ is a nonzero homogeneous linear equation on $\CC^\calF$.

   \item For any $x\in X$, $\{\nabla_\TT f(x) \mid f\in\CC^\calF\}$ is the linear span \defcolor{$\CC\calF$} of $\calF$.
  \end{enumerate}

  Suppose that the affine span of $\calF$ does not contain the origin. Then 
  \begin{enumerate}
  \setcounter{enumi}{2} 
  \item For any $f\in\CC^\calF$ and $x\in X$, $\nabla_\TT f=\bfz$ implies that $f(x)=0$.

   \item For any $x\in X$, the equation $\nabla_\TT f(x)=\bfz$ defines a linear subspace of $\CC^\calF$ of codimension
           $\dim \CC\calF$.

  \end{enumerate}
\end{Lemma}  
\begin{proof}
  Writing $f=\sum_{a\in\calF} c_a x^a$, the first statement is obvious.
  We have $\nabla_\TT f=\sum a c_a x^a$.
  As the coefficients $c_a$ are independent complex numbers and $x^a\neq 0$, Statement (2) is immediate.
  The hypothesis that the affine span of $\calF$ does not contain the origin implies that any $f\in\CC^\calF$ is
  quasi-homogeneous. 
  Statement (3) follows from Equation~\eqref{Eq:Euler}.
  The last statement follows from the observation that the set of $f$ such that $\nabla_\TT f=\bfz$ is the kernel of a surjective
  linear map $\CC^\calF\twoheadrightarrow \CC\calF$.
\end{proof}

Let $\defcolor{\calF}\vcentcolon=G\cap(\calA(\Gamma)\cup\{\bfe\})$, where $\bfe=(\bfz,1)$, be the (common) support of the diagonal
polynomials $f_i$ and let $\defcolor{\calG}\vcentcolon=G\cap\calA(\Gamma)$ be the (common) support of the polynomials $g_j$.
We either have that $\calF=\calG$ or $\calF=\calG\cup\{\bfe\}$.
Also, $|\calF|>1$ as $G$ is not a vertex, and as $G$ is a proper face of $Q=\conv(\calA(\Gamma)\cup\{\bfe\})$, but not its base, the
polynomials $f_i, g_j$ are quasi-homogeneous with a common quasi-homogeneity.

The parameter space for the entries of $(L_c(z)-\lambda I)|_G$ is 
\[
  \defcolor{Z}\ \vcentcolon=\ \bigl(\CC^\calF)^{\oplus m} \oplus \bigl(\CC^\calG)^{\oplus m}\,.
\]
We write $c=\defcolor{(f_\bullet, g_\bullet)}=(f_1,\dotsc,f_m, g_1,\dotsc,g_m)$ for points of $Z$.
This is a coordinate subspace of the parameter space $Y'$.
As $Z$ contains exactly those parameters that can appear in the facial polynomial $\psi_c(x)$, it suffices to show that the
set of parameters $c=(f_\bullet, g_\bullet)\in Z$ such that $\Var(\psi_c(x))$ is singular lies in a proper subvariety of $Z$.
The same case distinctions (i)---(v) in the proof of Theorem~\ref{Th:smoothFaces} apply.

After reindexing, Case (i) in the proof of Theorem~\ref{Th:smoothFaces} follows from the next lemma.

\begin{Lemma}\label{L:i}
  The set
  \[
  \defcolor{\Theta}\ \vcentcolon=\ \{c\in Z\mid \exists x\in X\mbox{ with }f_1(x)=f_2(x)=g_1(x)=g_2(x)=0\}
  \]
  lies in a proper subvariety of $Z$.
\end{Lemma}
\begin{proof}
  Consider the incidence correspondence,
  \[
  \defcolor{\Upsilon}\ \vcentcolon=\ \{(x,f_\bullet,g_\bullet)\in X\times Z \mid f_1(x)=f_2(x)=g_1(x)=g_2(x)=0\}\,.
  \]
  This has projections to $X$ and to $Z$ and its image in $Z$ is the set $\Theta$.

  Consider the projection $\pi_X\colon\Upsilon\to X$.
  By Lemma~\ref{L:linearEquations}(1), for $x\in X$, each condition $f_i(x)=0$, $g_i(x)=0$ for $i=1,2$
  is a linear equation on $\CC^\calF$ or $\CC^\calG$.
  These are independent on $Z$ as they involve different variables.
  Thus the fiber $\pi_X^{-1}(x)$ is a vector subspace of $Z$ of codimension 4, and 
  $\dim(\Upsilon)=\dim(Z)-4+\dim(X)=\dim(Z)+d-3$.

  Consider the projection $\pi_Z$ to $Z$ and let $(f_\bullet,g_\bullet)\in\pi_Z(\Upsilon)$.
  Then there is an $x\in X$ such that $f_1(x)=f_2(x)=g_1(x)=g_2(x)=0$.
  Let $w\in\ZZ^{d+1}$ be a common quasi-homogeneity of the polynomials $f_i,g_j$. 
  By Lemma~\ref{L:quasi-homogeneous} (1), for any $t\in\CC^\times$, each of $f_1,f_2,g_1,g_2$ vanishes at $t^w\cdot x$.
  Thus the fiber $\pi_Z^{-1}(f_\bullet,g_\bullet)$ has dimension at least one.
  By the Theorem~\cite[Theorem 1.25]{Shaf13} on the dimension of the image and fibers of a map, the image $\pi_Z(\Upsilon)$ has
  dimension at most $\dim(Z)+d-4<\dim(Z)$, which establishes the lemma.
\end{proof}  

After reindexing and possibly interchanging $f$ with $g$, Cases (ii) and (iii) in the proof of Theorem~\ref{Th:smoothFaces} follow from the
next lemma. 

\begin{Lemma}\label{L:ii}
  The set
  \[
  \Theta\ \vcentcolon=\ \{c\in Z\mid \exists x\in X\mbox{ with }f_1(x)=f_2(x)=g_1(x)=0\mbox{ and }\,\nabla_\TT g_1(x)=\bfz\}
  \]
  lies in a proper subvariety of $Z$.
\end{Lemma}
\begin{proof}
  Consider the incidence correspondence,
  \[
  \Upsilon\ \vcentcolon=\ \{(x,f_\bullet,g_\bullet)\in X\times Z \mid f_1(x)=f_2(x)=g_1(x)=0
    \mbox{ and }\nabla_\TT g_1(x)=\bfz\}\,.
  \]
  Let $x\in X$ and consider the fiber $\pi_X^{-1}(x)$.
  As in the proof of Lemma~\ref{L:i}, the conditions $f_1(x)=f_2(x)=0$ are two independent linear equations on $Z$.
  By Lemma~\ref{L:linearEquations} (3), $\nabla_\TT g_1(x)=\bfz$ implies that $g_1(x)=0$, and by 
  Lemma~\ref{L:linearEquations} (4), the condition $\nabla_\TT g_1(x)=\bfz$ is $\dim \CC\calG$ further independent linear equations on $Z$.

  If $|\calG|=1$, so that $g_1=c_ax^a$ is a single term, then $g(x)=0$ implies that $c_a=0$.
  Consequently, the image $\Theta$ of $\Upsilon$ in $Z$ lies in a proper subvariety.
  Otherwise, $|\calG|>1$ which implies that $\dim \CC\calG\geq 2$, and thus the fiber has codimension at least 4.
   As in the proof of Lemma~\ref{L:i}, this implies that $\Theta$ lies in a proper subvariety of $Z$.
\end{proof}

 Case (iv) in the proof of Theorem~\ref{Th:smoothFaces} is more involved.

\begin{Lemma}\label{L:iv}
  The set
  \[
  \defcolor{\Theta}\ \vcentcolon=\ \{c\in Z\mid \exists x\in X\mbox{ with }f_1(x)=g_1(x)=0\mbox{ and }\,\nabla_\TT \psi_c(x)=\bfz\}
  \]
  lies in a proper subvariety of $Z$.
\end{Lemma}
\begin{proof}
  The set $\Theta$ includes the sets of Lemmas~\ref{L:i} and~\ref{L:ii}.
  Let $\defcolor{\Theta^\circ}\subset\Theta$ be the set of $c=(f_\bullet,g_\bullet)$ that have a witness $x\in X$ 
  ($f_1(x)=g_1(x)=0$ and $\nabla_\TT \psi_c(x)=\bfz$) such that none of
  $\nabla_\TT f_1(x)$, $\nabla_\TT g_1(x)$, or $f_i(x)g_i(x)$ for $i>1$ vanish.
  It will suffice to show that $\Theta^\circ$ lies in a proper subvariety of $Z$.

  For this, we use the incidence correspondence,
  \begin{multline*}
   \qquad\defcolor{\Upsilon}\ \vcentcolon=\ \{(y,x,f_\bullet,g_\bullet)\in \CC^\times \times X\times Z \mid f_1(x)=g_1(x)=0\,,\\
   y\prod_{i=2}^m f_i(x)\ -\ (-1)^m\prod_{i=2}^m g_i(x)=0\,,\mbox{ and }
   \nabla_\TT f_1(x)\ -\ (-1)^my \nabla_\TT g_1(x)=\bfz\}\,.\qquad
  \end{multline*}
  We show that $\Theta^\circ\subset\pi_Z(\Upsilon)$.
  Let $c=(f_\bullet,g_\bullet)\in\Theta^\circ$ with witness $x\in X$ in that $f_1(x)=g_1(x)=0$ and $\nabla_\TT \psi_c(x)=\bfz$, but none of
  $\nabla_\TT f_1(x)$, $\nabla_\TT g_1(x)$, or $f_i(x)g_i(x)$ for $i>1$ vanish.
  There is a unique $\defcolor{y}\in\CC^\times$ satisfying
  \[
  y \prod_{i=2}^m f_i(x)\ -\ (-1)^m \prod_{i=2}^m g_i(x)\ =\ 0\,.
  \]
 Dividing~\eqref{Eq:iv_consequence}
  by $\prod_{i=2}^m f_i(x)$ gives $\nabla_\TT f_1(x)-(-1)^my \nabla_\TT g_1(x)=\bfz$, and thus
  $(y,x,f_\bullet,g_\bullet)\in\Upsilon$.

  We now determine the dimension of $\Upsilon$.
  Let $(y,x)\in\CC^\times \times X$ and consider the fiber $\pi^{-1}(y,x)\subset Z$ above it in $\Upsilon$.
  The two linear and one nonlinear equations
  \begin{equation}\label{Eq:Three}
  f_1(x)\ =\ g_1(x)\ =\  y\prod_{i=2}^m f_i(x)-(-1)^m\prod_{i=2}^m g_i(x)\ =\ 0
  \end{equation}
  are independent on $Z$ as they involve disjoint sets of variables, and thus 
  define a subvariety $\defcolor{T}\subset Z$ of codimension 3.
  Consider the remaining equation, $\nabla_\TT f_1(x)-(-1)^my \nabla_\TT g_1(x)=\bfz$.

  Note that if $\bfe=(\bfz,1)$ lies in the support $\calF$ of $f_1$, so that $\calF=\calG\cup\{\bfe\}$, then
  $\nabla_\TT f_1(x)$ contains the term $-\bfe$ and thus cannot lie in the span $\CC\calG$ of $\calG$, which contains $\nabla_\TT g_1(x)$
  by Lemma~\ref{L:linearEquations}(2).
  In this case the fiber is empty and $\Theta^\circ=\emptyset$.

  Suppose that $\calF=\calG$ and $(f_\bullet,g_\bullet)\in T$.
  Let $w\in\ZZ^{d+1}$ be any homogeneity for $f_1$ (or $g_1$).
  Then there exists  $w_\calF\neq 0$ such that $w\cdot a=w_\calF$ for all $a\in\calF$.
  Equation~\eqref{Eq:Euler} implies that
  \[
  w\cdot \nabla_\TT f_1(x)\  =\ w_\calF\, f_1(x)\ =\ 0\,,
  \]
  and the same for $g_1$.
  Thus $\nabla_\TT f_1(x)$ and $\nabla_\TT g_1(x)$ are annihilated by all homogeneities and so lie in the affine span of $\calF$---the linear span
  of differences $a{-}b$ for $a,b\in\calF$. 
  This has dimension $\dim \CC\calF-1$.
  Consequently, $\nabla_\TT f_1(x)-(-1)^my \nabla_\TT g_1(x)=\bfz$ consists of
  $\dim \CC\calF-1$ independent linear equations on the subset of
  $\CC^\calF\oplus\CC^\calF$ consisting of pairs $f_1,g_1$ such that $f_1(x)=g_1(x)=0$.
  These are independent of the third equation in~\eqref{Eq:Three}.
  Thus the fiber $\pi^{-1}(y,x)\subset Z$ has codimension $3+\dim\CC\calF-1=2+\dim\CC\calF$ and so 
  \[
  \dim\Upsilon\ =\ \dim(\CC^\times\times X)+\dim Z-\dim\CC\calF-2\ =\ \dim Z+d-\dim\CC\calF\,.
  \]
  
  Let $(f_\bullet,g_\bullet)\in\pi_Z(\Upsilon)$ have witness $(y,x)$.
  That is, the equations~\eqref{Eq:Three} hold, as well as $\nabla_\TT f_1(x)-(-1)^my \nabla_\TT g_1(x)=\bfz$.
  As in the proof of Lemma~\ref{L:i}, if $w\in\ZZ^{d+1}$ is a quasi-homogeneity for polynomials of support $\calF$, then
  $(y, t^w\cdot x)$ also satisfies these equations.

  We have $\calF=\calG= G\cap\calA(\Gamma)$, so that $G$ is a face of the base of $Q$.
  Thus there are at least two (in fact the codimension of $G$ in $Q$) independent homogeneities, which implies that the fiber
  $\pi_Z^{-1}(f_\bullet,g_\bullet)$ has dimension at least two.
  This implies that the image $\Theta^\circ$ has dimension at most $\dim Z+d-\dim\CC\calF-2$.
  Since $G$ is not a vertex, $\dim \CC\calF\geq 2$, which shows that $\dim\Theta^\circ<\dim Z$ and completes the proof.
\end{proof}  

%
\subsection{Case (v)}\label{S:Case(v)}

For $\alpha\in\CC^\times$, define $\defcolor{\Psi(\alpha,f_\bullet, g_\bullet)}\subset X$ to be the set
\[
   \Big\{x\in X \;\Big|\;  \mbox{none of $f_i(x)g_i(x)$ for $i\geq 1$ vanish}
     \mbox{ \ and \ } \prod_{i=1}^m f_i(x)\ -\ (-1)^m\alpha \prod_{i=1}^m g_i(x)\ =\ 0\Big\}\ .
\]
Case (v) in the proof of Theorem~\ref{Th:smoothFaces} follows from the next lemma.

\begin{Lemma}\label{L:A}
  There is a dense open subset $U_1\subset Z$ such that if $(f_\bullet,g_\bullet)\in U_1$, then
  $\Psi(1,f_\bullet,g_\bullet)$ is smooth.
\end{Lemma}

We will deduce this from a weaker lemma.

\begin{Lemma}\label{L:B}
  There is a dense open subset $U\subset \CC^\times\times Z$ such that if $(\alpha,f_\bullet,g_\bullet)\in U$, then
  $\Psi(\alpha,f_\bullet,g_\bullet)$ is smooth.
\end{Lemma}

\begin{proof}[Proof of Lemma~\ref{L:A}]
  If we knew that the set $U$ of Lemma~\ref{L:B} contained a point $(1,f_\bullet,g_\bullet)$, then
  $U_1\vcentcolon=U\cap(\{1\}\times Z)$ would be a dense open subset of $Z$,  which
  would complete the proof.
  As we do not know this, we must instead argue indirectly.

  Suppose that there is no such open set $U_1$ as in Lemma~\ref{L:A}.
  Then the set $\defcolor{\Xi}\subset Z$ consisting of $(f_\bullet,g_\bullet)$ such that $\Psi(1,f_\bullet,g_\bullet)$ is singular is dense
  in $Z$.

  For $\alpha\in\CC^\times$ and $(f_\bullet,g_\bullet)\in Z$, define \defcolor{$\alpha.(f_\bullet,g_\bullet)$} to be 
  $(f_\bullet,\alpha.g_\bullet)$ where
  \[
  \alpha.(g_1,g_2,\dotsc,g_m)\ =\ (\alpha g_1,g_2,\dotsc,g_m)\,.
  \]
  This is a $\CC^\times$-action on $Z$.
  Consequently, $\alpha.\Xi$ is dense in $Z$ for all $\alpha\in\CC^\times$.

  Let $U\subset\CC^\times\times Z$ be the set of Lemma~\ref{L:B}.
  As it is nonempty, let $(\alpha,f'_\bullet, g'_\bullet)\in U$.
  Then $\defcolor{U_\alpha}\vcentcolon= U\cap (\{\alpha\}\times Z)$ is nonempty and open in $\{\alpha\}\times Z$.
  As $\alpha.\Xi$ is dense, we have 
  \[
  U_\alpha \bigcap \big(\{\alpha\}\times \alpha.\Xi\big)\ \neq\ \emptyset\,.
  \]
  This is a contradiction, for if $(\alpha,f_\bullet,g_\bullet)\in U_\alpha$, then
  $\Psi(\alpha,f_\bullet, g_\bullet)$ is smooth, but if $(f_\bullet,g_\bullet)\in \alpha.\Xi$, then
  $(f_\bullet,\alpha^{-1}g_\bullet)\in\Xi$ and $\Psi(1,f_\bullet,\alpha^{-1}g_\bullet)$ is singular.
  The contradiction follows from the equality of sets $\Psi(\alpha,f_\bullet, g_\bullet)=\Psi(1,f_\bullet,\alpha^{-1}g_\bullet)$.
\end{proof}

\begin{proof}[Proof of Lemma~\ref{L:B}]
  Let $\defcolor{T}\subset X\times Z$ be the set of $(x,f_\bullet,g_\bullet)$ such that none of $f_i(x)g_i(x)$ for $i\geq 1$ vanish.
  Define $\varphi\colon T\to \CC^\times\times Z$ by
  \[
  \varphi(x,f_\bullet,g_\bullet)\ =\
  \bigl((-1)^m{\textstyle \prod_{i=1}^m f_i(x)/ \prod_{i=1}^mg_i(x)} \;,\; f_\bullet\,,\,g_\bullet \bigr)\,.
  \]
  Notice that $\varphi^{-1}(\alpha,f_\bullet,g_\bullet)=\Psi(\alpha,f_\bullet,g_\bullet)$ for
  $(\alpha,f_\bullet,g_\bullet)\in\CC^\times\times Z$.

  We claim that $\varphi(T)$ is dense in $\CC^\times\times Z$.
  For this, recall that the polynomials $f_i$ have support $\calF$, which is $G\cap(\calA(\Gamma)\cup\{\bfe\})$
  for some face $G$ of $Q=\conv(\calA(\Gamma)\cup\{\bfe\})$ that is neither its base nor a vertex, and the 
  polynomials $g_i$ have support $\calG=G\cap\calA(\Gamma)$.
  Since $G$ is not a vertex, there are $a,b\in\calF$ with $a\neq b$ and $b\in\calA(\Gamma)$.

  Let $\defcolor{f_i}\vcentcolon=x^a$ and $\defcolor{g_i}\vcentcolon=x^b$ for $i=1,\dotsc,m$.
  Then $X\times\{(f_\bullet,g_\bullet)\}\subset T$ and for $x\in X$
  $\varphi(x,f_\bullet,g_\bullet)=(x^{ma}-(-1)^mx^{mb},f_\bullet,g_\bullet)$.
  The map $X=(\CC^\times)^{d+1}\to\CC^\times$ given by $x\mapsto x^{ma}-(-1)^mx^{mb}$ is surjective as $ma-mb\neq 0$.
  This implies that the differential $d\varphi$ is surjective at any point of $X\times\{(f_\bullet,g_\bullet)\}$, and therefore $\varphi(T)$ 
  is dense in $\CC^\times\times Z$.

  Since $T$ is an open subset of the smooth variety $X\times Z$, it is smooth.
  Then Bertini's Theorem~\cite[Thm.\ 2.27, p.\ 139]{Shaf13} implies that there is a dense open subset
  $U\subset \CC^\times\times Z$ such that for
  $(\alpha,f_\bullet,g_\bullet)\in U$,
  $\varphi^{-1}(\alpha,f_\bullet,g_\bullet)=\Psi(\alpha,f_\bullet,g_\bullet)$ is smooth.
\end{proof}

\section{Critical points property}\label{S:final}
We illustrate our results, using them to establish the critical points property (and thus the spectral edges nondegeneracy conjecture) for
three periodic graphs.
We first state this property.

Let $\Gamma$ be a connected $\ZZ^d$-periodic graph with parameter space $Y=\CC^E\times\CC^W$ for discrete
operators on $\Gamma$.
We say that $\Gamma$ has the \demph{critical points property} if there
is a dense open subset $U\subset Y$ such that if $c\in U$, then every critical point of the function $\lambda$ on the Bloch variety
$\Var(D_c(z,\lambda))$ is nondegenerate in that the Hessian determinant
 \begin{equation}\label{Eq:Hessian}
   \det \left( \left( \frac{\partial^2\lambda}{\partial z_i\partial z_j}\right)_{i,j=1}^{d} \right)
 \end{equation}
is nonzero at that critical point.
Here, the derivatives are implicit, using that $D(z,\lambda)=0$.  

\subsection{Reformulation of Hessian condition}
Let $D=\det(L_c(z)-\lambda I)$ be the dispersion function for an operator $L_c$ on a periodic graph $\Gamma$.
In Section~\ref{S:A} we derived the equations for the critical points of the function $\lambda$ on the Bloch variety
$\Var(D(z,\lambda))$,
 \begin{equation}\label{Eq:CPE_nonToric}
     D(z,\lambda)\ =\ 0
    \qquad\mbox{and}\qquad
  \frac{\partial D}{\partial z_i}\ =\ 0 \qquad\mbox{for } i=1,\dots,d\,.
 \end{equation}
Implicit differentiation of $D=0$ gives
$\frac{\partial D}{\partial z_j}+\frac{\partial D}{\partial\lambda}\cdot\frac{\partial \lambda}{\partial z_j}=0$.
If $\frac{\partial D}{\partial\lambda}\neq 0$, then $\frac{\partial\lambda}{\partial z_j}=0$.
If $\frac{\partial D}{\partial\lambda}=0$, then $(z,\lambda)$ is a singular point hence is also a critical point of the function $\lambda$
and so we again have $\frac{\partial\lambda}{\partial z_j}=0$.
Differentiating again we obtain,
\[
0\ =\
   \frac{\partial}{\partial z_i}\left(\frac{\partial D}{\partial z_j}
   +\frac{\partial D}{\partial\lambda}\cdot\frac{\partial \lambda}{\partial z_j}\right)
 \ =\ \frac{\partial^2D}{\partial z_i\partial z_j} + 
      \frac{\partial^2D}{\partial z_i\partial\lambda}\cdot\frac{\partial \lambda}{\partial z_j} +
      \frac{\partial D}{\partial\lambda}\cdot \frac{\partial^2\lambda}{\partial z_i\partial z_j}\,.
\]
At a critical point (so that $\frac{\partial\lambda}{\partial z_j}=0$), we have
\[
\frac{\partial^2D}{\partial z_i\partial z_j} \ =\
-\frac{\partial D}{\partial\lambda}\cdot \frac{\partial^2\lambda}{\partial z_i\partial z_j}\,.
\]
Thus 
\[
    \det \left( \left( \frac{\partial^2 D}{\partial z_i\partial z_j}\right)_{i,j=1}^{d} \right)
    \ =\ \left(-\frac{\partial D}{\partial\lambda}\right)^d\cdot
    \det \left( \left( \frac{\partial^2\lambda}{\partial z_i\partial z_j}\right)_{i,j=1}^{d} \right)\ .
\]
Consider now the Jacobian matrix of the critical point equations~\eqref{Eq:CPE_nonToric},
\[
  J\ =\ \left(\begin{array}{cccc}
     \frac{\partial D}{\partial z_1} & \dotsc &  \frac{\partial D}{\partial z_d} & \frac{\partial D}{\partial\lambda} \\
     \frac{\partial^2 D}{\partial z_1^2} & \dotsc &  \frac{\partial^2 D}{\partial z_d\partial z_1}      \rule{0pt}{14pt}
     & \frac{\partial^2 D}{\partial\lambda\partial z_1} \\
     \vdots & \ddots & \vdots & \vdots\\
     \frac{\partial^2 D}{\partial z_1\partial z_d} & \dotsc &  \frac{\partial^2 D}{\partial z_d^2}
     & \frac{\partial^2 D}{\partial\lambda\partial z_d} 
   \end{array}\right)\ .
\]
At a critical point, the first row is $(0\; \dotsb\; 0\; \frac{\partial D}{\partial\lambda})$, and thus
\[
   \det(J)\ =\ \frac{\partial D}{\partial\lambda}
               \det \left( \left( \frac{\partial^2 D}{\partial z_i\partial z_j}\right)_{i,j=1}^{d} \right)
          \ =\ 
                (-1)^d \left(\frac{\partial D}{\partial\lambda}\right)^{d+1}\cdot
              \det \left( \left( \frac{\partial^2\lambda}{\partial z_i\partial z_j}\right)_{i,j=1}^{d} \right)\ .
\]

A solution $\zeta$ of a system of polynomial equations on $\CC^n$ is \demph{regular} if the Jacobian of the system at $\zeta$ has full
rank $n$. 
Regular solutions are isolated and have multiplicity 1.
We deduce the following lemma.

\begin{Lemma}\label{L:nondegenerate}
  A nonsingular critical point $(z,\lambda)$ on $\Var(D_c(z,\lambda))$ is nondegenerate if and only if it is a regular solution 
  of the critical point equations~\eqref{Eq:CPE_nonToric}.
\end{Lemma}

The following theorem is adapted from arguments in~\cite[Sect.\ 5.4]{DKS}.

\begin{Theorem}\label{Th:singleCalculation}
  Let $\Gamma$ be a $\ZZ^d$-periodic graph.
  If there is a parameter value $c\in Y$ such that the critical point equations have $(d{+}1)!\vol(\calN(\Gamma))$ regular solutions, then the
  critical points property holds for $\Gamma$.
\end{Theorem}
\begin{proof}
  Let $Y$ be the parameter space for operators $L$ on $\Gamma$.
  Consider the variety
  \[
    \defcolor{CP}\ \vcentcolon=\
    \{(c,z,\lambda)\in Y\times (\CC^\times)^d\times\CC \mid \
    \mbox{the critical point equations~\eqref{Eq:CPE} hold}\}\,,
  \]
  which is the incidence variety of critical points on all Bloch varieties for operators on $\Gamma$.
  Let $\pi$ be its projection to $Y$.
  For any $c\in Y$, the fiber $\pi^{-1}(c)$ is the set of critical points of the function $\lambda$ on the corresponding Bloch variety for $D_c$.
  By Corollary~\ref{C:Bound}, there are at most $(d{+}1)!\vol(\calN(D_c))$ isolated points in the fiber.
  
  Let $c\in Y$ be a point such that the critical point equations have $(d{+}1)!\vol(\calN(\Gamma))$ regular solutions.
  Then $(d{+}1)!\vol(\calN(\Gamma))\leq(d{+}1)!\vol(\calN(D_c))$.
  By Lemma~\ref{L:NewtonPolytope}, $\calN(D_c)$ is a subpolytope of $\calN(\Gamma)$, so that $\vol(\calN(D_c))\leq\vol(\calN(\Gamma))$.
  We conclude that both polytopes have the same volume and are therefore equal.
  In particular, the corresponding Bloch variety has the maximum number of critical points, and each is a regular solution of the critical point
  equations~\eqref{Eq:CPE}.
  Because they are regular solutions, the implicit function theorem implies that there is a neighborhood \defcolor{$U_c$} of $c$ in the
  classical topology on $Y$ such that the map $\pi^{-1}(U_c)\to U_c$ is proper (it is a $(d{+}1)!\vol(\calN(\Gamma))$-sheeted cover).

  The set \defcolor{$DC$} of degenerate critical points is the closed subset of $CP$ given by the vanishing of the Hessian
  determinant~\eqref{Eq:Hessian}. 
  Since $\pi$ is proper over $U_c$, if $\defcolor{DP}=\pi(DC)$ is the image of $DC$ in $Y$, then $DP\cap U_c$ is closed in $U_c$.
  As the points of $\pi^{-1}(c)$ are regular solutions, Lemma~\ref{L:nondegenerate} implies they are all nondegenerate and thus
  $c\not\in DP$, so that $U_c\smallsetminus DP$ is a nonempty classically open subset of $Y$ consisting of parameter values
  $c'$ with the property that all critical points on the corresponding Bloch variety are nondegenerate.

  This implies that there is a nonempty Zariski open subset of $Y$ consisting of parameters such that all critical points on the
  corresponding Bloch variety are nondegenerate, which completes the proof.
\end{proof}  

By Theorem~\ref{Th:singleCalculation}, it suffices to find a single Bloch variety with the maximum number of isolated critical points to
establish the critical points property for a periodic graph.
The following examples use such a computation to establish the critical points property for $2^{19}+2$ graphs $\Gamma$.
Computer code and output are available at the github repository\footnote{\url{https://mattfaust.github.io/CPODPO}.}.

\begin{Example}\label{Ex:DenseGraph}
  Let us consider the dense $\ZZ^2$-periodic graph $\Gamma$ of Figure~\ref{F:New-dense}.
  It has $m=2$ points in its fundamental domain and the convex hull of its support $\calA(\Gamma)$ has area 4.
  By Theorem~C, a general operator on $\Gamma$ has $2!\cdot 2^{2+1}\cdot 4 = 64$ critical points.
\begin{figure}[htb]
  \centering
   \includegraphics{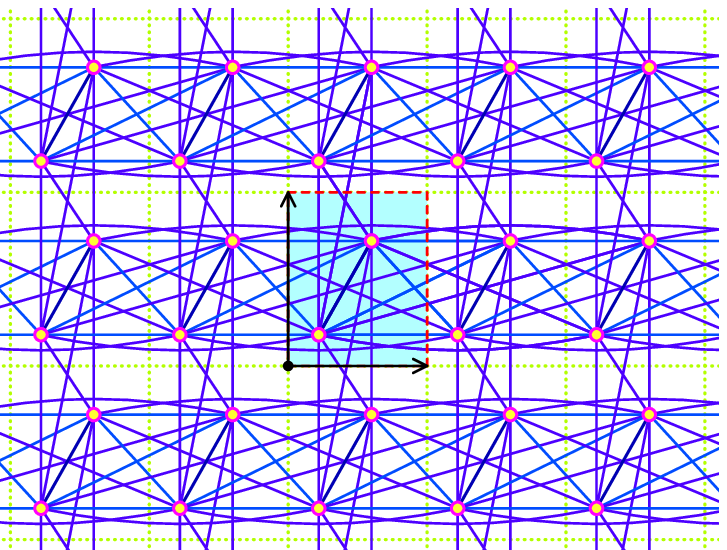}
     \qquad
  \begin{picture}(212,156)(-2,-25)
    \put(0,0){\includegraphics[height=110pt]{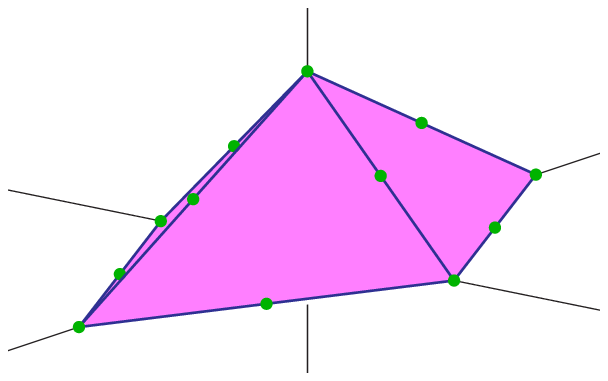}}

     \thicklines

     \thinlines   
      \put( 56,91){\footnotesize$(0,0,2)$}
 
      \put(161, 52){\footnotesize$(-4,0,0)$}
 
       \put(-1,39){\footnotesize$(0,-2,0)$}
       \put(14, 3.5){\footnotesize$(4,0,0)$}  \put(118, 16){\footnotesize$(0,2,0)$}

 \end{picture}

   \caption{Dense periodic graph and its polytope from Figure~\ref{F:New-dense}.}
   \label{F:newDensewidePyramid}
 \end{figure}
There are 13 edges and two vertices in $W$, and independent computations in the computer algebra systems 
Macaulay2~\cite{M2} and Singular~\cite{Singular} find a point $c\in Y=\CC^{15}$
  such that the critical point equations have 64 regular solutions on $(\CC^\times)^2\times\CC$.
  By Theorem~\ref{Th:singleCalculation}, the critical points property holds for $\Gamma$.
  These computations are independent in that the code, authors, and parameter values for each are distinct.\hfill$\diamond$
\end{Example}

\begin{Example}\label{Ex:FullNotDense}
  The graph $\Gamma$ in Figure~\ref{F:FullNotDense} is not dense.
  Its restriction to the fundamental domain is not the
  complete graph on 3 vertices and there are three and not nine edges between any two adjacent translates of the fundamental domain.
  Altogether, it has $3\cdot 6 + 1=19$ fewer edges than the corresponding dense graph.
  Its support $\calA(\Gamma)$ forms the columns of the matrix
  $(\begin{smallmatrix}0&1&1&0&-1&-1&0\\0&0&1&1&0&-1&-1\end{smallmatrix})$ whose convex hull is a hexagon of area 3.

  \begin{figure}[htb]
  \centering
  \includegraphics{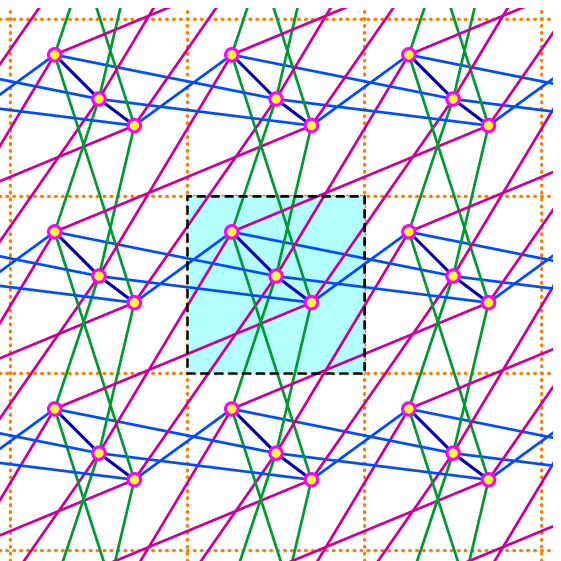}
  \qquad
   \begin{picture}(235,160)(-10,0)
     \put(0,0){\includegraphics[height=150pt]{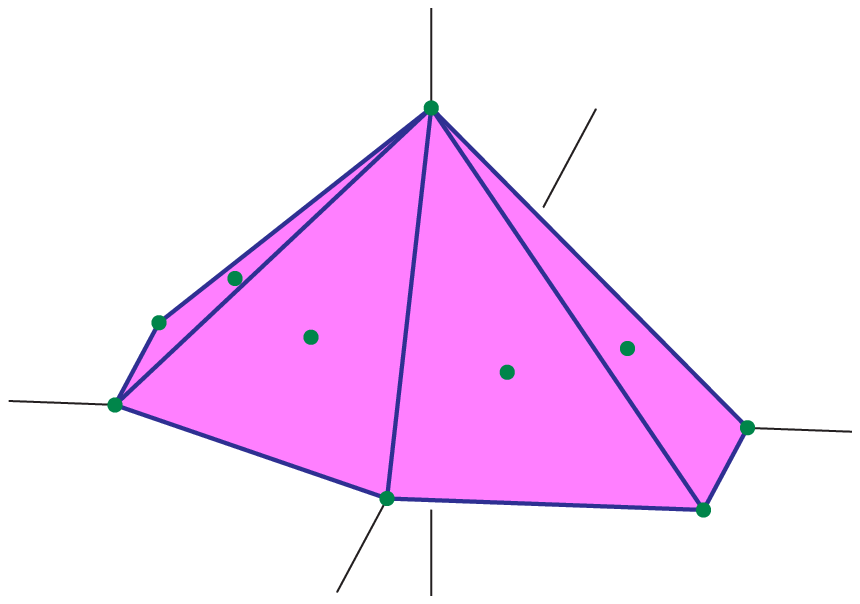}}
     \thicklines
     \put(127,10){{\color{white}\line(0,1){41}}}   \put(128,10){{\color{white}\line(0,1){41}}}
     \put(39.5,10){{\color{white}\line(2,3){33}}}  \put(40.5,10){{\color{white}\line(2,3){33}}}
     \put(178.5,104.5){{\color{white}\line(-1,-2){17.5}}}  \put(179.5,104.5){{\color{white}\line(-1,-2){17.5}}}
     \thinlines   
     \put( 72,125){\footnotesize$(0,0,3)$} 
     \put( 10,120.5){\footnotesize$(-1,-1,1)$}  \put( 39.5,117.5){\vector(1,-2){16}}
     \put(164,107.5){\footnotesize$( 1,2,1)$}   \put(179,104.5){\vector(-1,-2){19.5}}
     \put(-13, 69){\footnotesize$(-3,-3,0)$} 
     \put(-12, 38){\footnotesize$(0,-3,0)$}  \put(190, 48){\footnotesize$(0,3,0)$}
     \put(59.5, 16.5){\footnotesize$(3,0,0)$}   \put(162, 10){\footnotesize$(3,3,0)$}

    \put(14, 0){\footnotesize$(1,-1,1)$} \put(40,10){\vector(2,3){35.5}}
    \put(111.5,0){\footnotesize$(2,1,1)$} \put(127.5,10){\vector(0,1){43.5}}

   \end{picture}
  \caption{Sparse graph with the same Newton polytope as the corresponding dense graph.}
  \label{F:FullNotDense}
  \end{figure}

  Despite $\Gamma$ not being dense, its Newton polytope $\calN(\Gamma)$ is equal to the Newton polytope of the dense graph with the same
  parameters, $\calA(\Gamma)$ and $W$.
  Figure~\ref{F:FullNotDense} displays the Newton polytope, along with elements of the support of the dispersion function
  that are  visible.
  Observe that on each triangular face, there are four and not ten monomials.
  
  By Theorem~A (Corollary~\ref{C:Bound}), there are at most $2!\cdot 3^{2+1}\cdot 3=162$ critical points.
  There are eleven edges and three vertices in $W$, and independent computations in Macaulay2 and Singular find a point
  $c\in Y=\CC^{14}$ such that the critical point equations have 162 regular solutions on $(\CC^\times)^2\times\CC$.
  By Theorem~\ref{Th:singleCalculation}, the critical points property holds for $\Gamma$.

  Let $\Gamma'$ be a graph that has the same vertex set and support as $\Gamma$, and contains all the edges of
  $\Gamma$---then~\cite[Thm.\ 22]{DKS} implies that the
  critical points property also holds for $\Gamma'$.
  This establishes the critical points property for an additional $2^{19}-1$ periodic graphs.\hfill$\diamond$
\end{Example}

  \begin{Example}\label{Ex:NotPyramid}
   The graph $\Gamma$ of Figure~\ref{F:notPyramidGraphPolytope} has 
   only ten edges but the same fundamental domain $W$ and support $\calA(\Gamma)$ as the the graph of Figure~\ref{F:FullNotDense},
   which had eleven edges.
   Its Newton polytope is smaller, as it is missing the vertices $(3,3,0)$ and $(-3,-3,0)$.
   \begin{figure}[htb]
  \centering
   \includegraphics{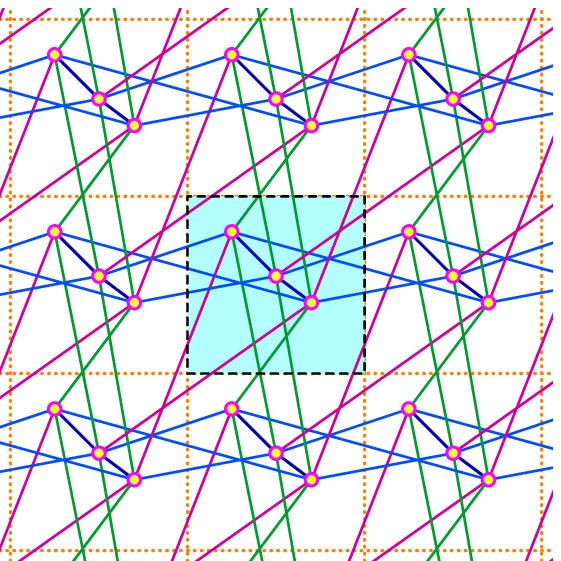}
     \qquad
  \begin{picture}(245,159)(-8,-2)
    \put(0,0){\includegraphics[height=150pt]{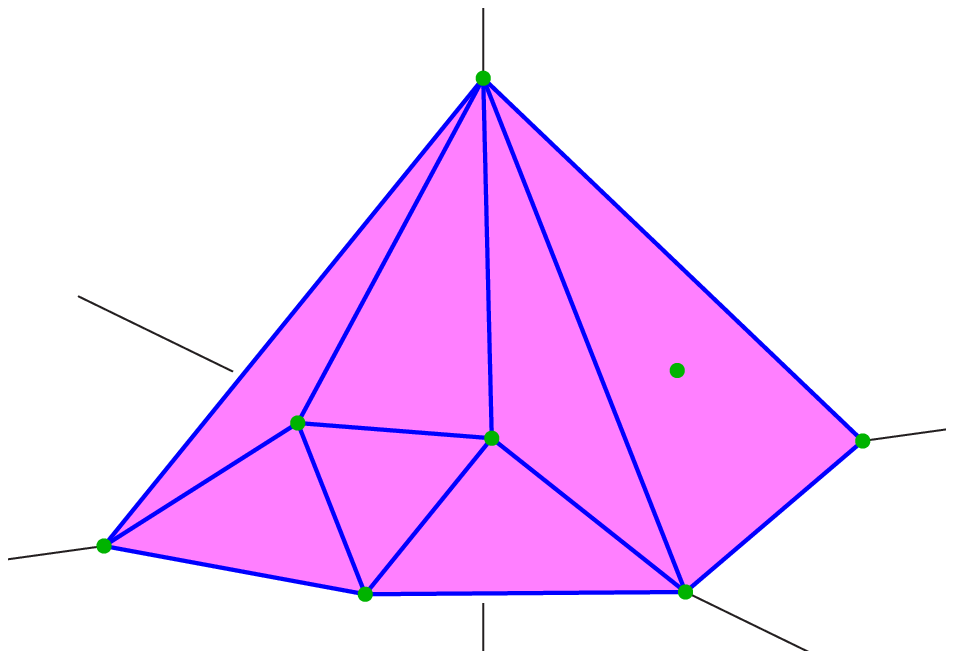}}

    \thicklines
     \put( 38.5,10){{\color{white}\line(2,3){24.5}}}       \put(39.5,10){{\color{white}\line(2,3){24.5}}}
     \put(178.5,110.5){{\color{white}\line(-1,-2){20}}}  \put(179.5,110.5){{\color{white}\line(-1,-2){20}}}
     \put(133, 8){{\color{white}\line(-1,2){17}}} \put(134, 8){{\color{white}\line(-1,2){17}}}
    \thinlines   
    \put( 75,132){\footnotesize$(0,0,3)$}
    \put(155.5,113.5){\footnotesize$(-1,1,1)$}   \put(179,110.5){\vector(-1,-2){21}}

    \put(197, 56){\footnotesize$(-3,0,0)$} 
    \put( 21,  0){\footnotesize$(2,1,1)$} \put( 39,10){\vector(2,3){26.5}}
    \put( 67,  2){\footnotesize$(2,2,0)$}  
    \put(120, -2){\footnotesize$(1,2,1)$}     \put(133.5, 8){\vector(-1,2){19}}

    \put(-11, 28){\footnotesize$(3,0,0)$}   \put(166, 12.5){\footnotesize$(0,3,0)$}
  \end{picture}
   \caption{A periodic graph and its Newton polytope.}
   \label{F:notPyramidGraphPolytope}
\end{figure}

   It has volume $70/3$ and normalized volume $3!\cdot 70/3=140$.
   Independent computations in Macaulay2 and Singular find a point $c\in Y=\CC^{13}$ such that the critical point equations have
   140 regular solutions on $(\CC^\times)^2\times\CC$.
   Thus there are no critical points at infinity, and Theorem B implies that the Bloch variety is smooth at infinity.

   As before, achieving the bound of Corollary~\ref{C:Bound} with regular solutions implies that all critical points are
   nondegenerate and the critical points property holds for $\Gamma$.  \hfill$\diamond$
  \end{Example}

\section{Conclusion}
We considered the critical points of the complex Bloch variety for an operator on a periodic graph.
We gave a bound on the number of critical points---the normalized volume of a Newton polytope---together with a criterion
for when that bound is attained.
We presented a class of graphs (dense periodic graphs) and showed that this criterion holds for general discrete operators on a
dense graph.
Lastly, we used these results to find $2^{19}+2$ graphs on which the spectral edges conjecture holds for general discrete operators
when $d=2$.


\bibliographystyle{amsplain}
\def\cprime{$'$}
\providecommand{\bysame}{\leavevmode\hbox to3em{\hrulefill}\thinspace}
\providecommand{\MR}{\relax\ifhmode\unskip\space\fi MR }
\providecommand{\MRhref}[2]{%
  \href{http://www.ams.org/mathscinet-getitem?mr=#1}{#2}
}
\providecommand{\href}[2]{#2}

\end{document}